\documentclass{article}

\usepackage{microtype}
\usepackage{graphicx}
\usepackage{subfigure}
\usepackage{booktabs} %

\usepackage{hyperref}

\usepackage[accepted]{icml2024}

\usepackage{amsmath}
\usepackage{amssymb}
\usepackage{mathtools}
\usepackage{amsthm}

\usepackage{booktabs}       %
\usepackage{amsfonts}       %
\usepackage{nicefrac}       %
\usepackage{microtype}      %
\usepackage{xcolor}         %

\newcommand{\Real}{\mathbb{R}}
\renewcommand{\subset}{\subseteq}
\usepackage{graphicx}
\usepackage{subcaption}

\usepackage{mathtools}
\usepackage{mathrsfs}
\usepackage{dsfont}

\usepackage{amsthm}

\usepackage[capitalize,noabbrev]{cleveref}

\newtheorem{thm}{\bf Theorem}[section]
\newtheorem{cor}[thm]{\bf Corollary}
\newtheorem{lem}[thm]{\bf Lemma}
\newtheorem{rem}[thm]{\bf Remark}

\newtheorem{prop}[thm]{\bf Proposition}

\newtheorem{dfn}[thm]{\bf Definition}

\newtheorem{hyp}[thm]{\bf Hypothesis}

\newcommand{\eps}{\varepsilon}

\newcommand{\tta}{\vartheta}
\newcommand{\R}{\mathbb{R}}
\newcommand{\N}{\mathbb{N}}

\newcommand{\uu}{\mathcal{U}}
\newcommand{\lsim}{\raisebox{-0.13cm}{~\shortstack{$<$ \\[-0.07cm]
      $\sim$}}~}
\newcommand{\lk}{L}
\newcommand{\xk}{x}
\newcommand{\uk}{u}
\newcommand{\dd}{\Omega}

\newcommand{\zero}{\mathbf{0}}

\newcommand{\ra}{\rightarrow}

\newcommand{\normr}[1]{\left\|#1\right\|_{\operatorname{R}}}

\newcommand{\norm}[1]{\left| #1 \right|}
\newcommand{\abs}[1]{\left\vert #1 \right\vert}

\newcommand{\norms}[1]{|#1|_{\infty}}

\newcommand{\oball}{\mathcal{B}}

\newcommand{\Lip}{\operatorname{Lip}}
\newcommand{\hv}{\hat{V}}
\newcommand{\hk}{\hat{\kappa}}
\newcommand{\tG}{\tilde{G}}
\newcommand{\tV}{\tilde{V}}
\newcommand{\tk}{\tilde{\kappa}}
\newcommand{\rr}{\mathcal{R}}
\newcommand{\ff}{\mathcal{F}}

\newcommand{\mvf}[1]{\left\|#1\right\|_{\operatorname{M, \infty}}}
\newcommand{\mv}[1]{\left\|#1\right\|_{\operatorname{M}}}

\DeclarePairedDelimiter\roundBracket{\lparen}{\rparen}
\newcommand{\round}{\roundBracket*}

\DeclarePairedDelimiter{\setBracket}{\{} {\}}
\newcommand{\set}{\setBracket*}

\NewDocumentCommand{\dist}{o}{\IfNoValueTF{#1} {\operatorname{d}} {\operatorname{d}\round{#1}} }
\newcommand{\ball}{\mathcal B}
\newcommand{\ballbrack}{\round}
\NewDocumentCommand{\B}{o m m}{\IfNoValueTF {#1} {\ball_{#2} \ballbrack {#3}} {\ball^{#1}_{#2} \ballbrack {#3}}}

\usepackage[textsize=tiny]{todonotes}

\icmltitlerunning{Physics-Informed Neural Network Policy Iteration}

\begin{document}

\twocolumn[
\icmltitle{Physics-Informed Neural Network Policy Iteration: \\
Algorithms, Convergence, and Verification}

\icmlsetsymbol{equal}{*}

\begin{icmlauthorlist}
\icmlauthor{Yiming Meng}{uiuc,equal}
\icmlauthor{Ruikun Zhou}{uw,equal}
\icmlauthor{Amartya Mukherjee}{uw}
\icmlauthor{Maxwell Fitzsimmons}{uw}
\icmlauthor{Christopher Song}{uw}
\icmlauthor{Jun Liu}{uw}
\end{icmlauthorlist}

\icmlaffiliation{uiuc}{University of Illinois at Urbana-Champaign, United States}
\icmlaffiliation{uw}{University of Waterloo, Waterloo, Canada}

\icmlcorrespondingauthor{Jun Liu}{j.liu@uwaterloo.ca}

\icmlkeywords{Machine Learning, ICML}

\vskip 0.3in
]

\printAffiliationsAndNotice{\icmlEqualContribution} %

\begin{abstract}
Solving nonlinear optimal control problems is a challenging task, particularly for high-dimensional problems. We propose algorithms for model-based policy iterations to solve nonlinear optimal control problems with convergence guarantees. The main component of our approach is an iterative procedure that utilizes neural approximations to solve linear partial differential equations (PDEs), ensuring convergence. We present two variants of the algorithms. The first variant formulates the optimization problem as a linear least square problem, drawing inspiration from extreme learning machine (ELM) for solving PDEs. This variant efficiently handles low-dimensional problems with high accuracy. The second variant is based on a physics-informed neural network (PINN) for solving PDEs and has the potential to address high-dimensional problems. We demonstrate that both algorithms outperform traditional approaches, such as Galerkin methods, by a significant margin. We provide a theoretical analysis of both algorithms in terms of convergence of neural approximations towards the true optimal solutions in a general setting. Furthermore, we employ formal verification techniques to demonstrate the verifiable stability of the resulting controllers.

\end{abstract}

\section{Introduction}

Reinforcement learning in discrete environments has achieved remarkable success over the past decade, from AlphaGo \cite{silver2017mastering} to the recent breakthrough of GPT-3 \cite{ouyang2022training}, which uses reinforcement learning \cite{schulman2017proximal} from human feedback to fine-tune large language models. However, reinforcement learning in continuous environments, where states and actions evolve continuously in both space and time, remains a challenge \cite{duan2016benchmarking}. Theoretically speaking, when considering continuous-time scenarios, the discrete-time Bellman equation is replaced by a nonlinear partial differential equation (PDE) known as the Hamilton–Jacobi–Bellman (HJB) equation. Solving and analyzing this equation, in general, becomes a complex task due to its intricate nature. One major challenge arises from the possibility that the optimal cost function may not be differentiable, even for relatively straightforward problems \cite{bertsekas2015value}. In such cases, one has to resort to viscosity solutions \cite{crandall1984some} to study HJB equations.

A rich literature exists on policy iteration techniques for obtaining suboptimal solutions to the HJB equations \cite{leake1967construction,saridis1979approximation,beard1995improving,beard1997galerkin,beard1998approximate}. One notable approach is to construct successive approximations to solutions of the so-called Generalized Hamilton-Jacobi-Bellman (GHJB) equation, which is a linear PDE and potentially easier to solve. Galerkin approximations for solving the GHJB are proposed in \cite{beard1997galerkin,beard1998approximate} and have proven effective for solving low-dimensional problems. However, such approaches do not scale well to high-dimensional problems. Indeed, Galerkin methods are known to suffer from the curse of dimensionality.

Motivated by recent successes in solving PDEs using neural networks \cite{raissi2019physics,huang2006extreme,chen2022bridging,han2018solving,sirignano2018dgm,weinan2021algorithms} and the potential of neural networks to overcome the curse of dimensionality \cite{poggio2017and}, we set out to revisit the policy iteration approach by solving GHJB equations using neural networks. The main goal is to answer the following questions:
\begin{enumerate}
    \item Can neural approximations of the solutions to GHJB converge to the viscosity solution of the HJB equation? 
\item Can neural approximations efficiently compute solutions of the HJB with high accuracy?
\item Can neural policy iteration overcome the curse of dimensionality?
\item Can neural approximations be guaranteed to lead to stabilizing controllers?  
\end{enumerate}

The answers to these questions are all positive to some degree. The main contributions of this paper are as follows: 
\begin{enumerate}
    \item We prove that policy iteration indeed converges to viscosity solutions of the HJB equation. 

\item We propose two variants of neural policy iteration. The first, inspired by the Extreme Learning Machine \cite{huang2006extreme} and termed ELM-PI, can achieve remarkable accuracy and efficiency on low-dimensional problems. The second, based on Physics-Informed Neural Networks (PINN) \cite{raissi2019physics} for solving PDEs, has been shown to scale better than ELM-PI as dimensions increase.

\item We formulate formal verification problems for the resulting controllers to verify their stability. We show with a simple example that seemingly convergent results can lead to unstable controllers, which necessitate the use of formal verification when safety is a concern.

\end{enumerate}

\textbf{Related work:} 
(1) The idea of policy iteration in the context of designing optimal stabilizing controllers has a long history. For linear systems, this reduces solving an algebraic Riccati equation (ARE), which is quadratic in the unknown matrix, into a sequence of Lyapunov equations, which are linear and easier to solve. The algorithm is known as Kleinman's algorithm \cite{kleinman1968iterative}. In the nonlinear case, this procedure reduces the nonlinear HJB equation to a sequence of linear PDEs that characterize the value (and Lyapunov) functions for the stabilizing controllers obtained at each iteration of policy evaluation. This result dates back at least to 1960s \cite{milshtein1964successive,vaisbord1963concerning} and followed by \cite{leake1967construction,saridis1979approximation,beard1995improving,jiang2017robust,bhasin2013novel,vrabie2009neural,jiang2012computational,jiang2014robust} and many others. To the best knowledge of the authors, however, none of these works establish the convergence of policy iteration to viscosity solutions \cite{crandall1984some} of the HJB equation, especially when function approximators are involved. Furthermore, classical computational approaches, such as Galerkin methods, for solving PDEs often do not scale well.

(2) We draw significant inspiration from recent work on neural networks for solving PDEs \cite{raissi2019physics, huang2006extreme, chen2022bridging, han2018solving, sirignano2018dgm} (see the recent survey \cite{weinan2021algorithms} and discussions on the potential for machine learning to overcome the curse of dimensionality when solving PDEs). To the best of our knowledge, no previous work has reported the use of neural PDE solving for policy iterations in the context of nonlinear optimal control on benchmark problems. We address these gaps in this paper.

\section{Problem formulation}
\label{sec:prob}

We consider a class of  optimal control problems subject to control-affine dynamical systems of the form
\begin{equation}\label{E:sys}
\dot{x}=f(x)+g(x)u,
\end{equation}
where $f:\mathbb{R}^n\rightarrow\R^n$ is a continuously differentiable  vector field and $g:\mathbb{R}^n\rightarrow\R^{n\times m}$ is smooth,  $x\in\R^n$ is the state, $u\in\R^m$ is the control input. We also assume that $f(\zero)=\zero$. 

We are interested in the case where the maximal interval of existence is $[0,\infty)$ for admissible controls. For simplicity, in the context of infinite-horizon trajectory, we overload the notation $u$ as the control signal,   i.e. $u:[0,\infty)\ra\R^m$.  Subject to the control  $u$, the unique solution starting from $x_0$  is denoted by $\phi(t; x_0, u)$. We may also write the solution as $\phi(t)$ or $\phi$, if the rest of the arguments are not emphasized. 

Let $R$ be a symmetric and positive definite matrix. Introduce $ \lk(\xk, \uk) = Q(\xk)+\normr{\uk}^2$, where $Q:\,\R^n\ra\R$ is a symmetric and positive definite function, %
and $\normr{\uk}=\uk^TR\uk$, 
where $R:\,\R^n\ra\Real^{m\times m}$ is also  symmetric and positive definite. The associated cost is then commonly defined as:
\begin{equation}\label{E:cost}
    J(x, u)=\int_0^\infty \lk(\phi(s; x, u), u(s))ds.
\end{equation}

\begin{dfn}[Admissible Controls]\label{dfn: adm}
Given a subset $\dd\subseteq\R^n$ containing the origin. A control $u:\Omega\ra\R^m$ is admissible on $\dd$, denoted as $u\in\uu(\dd)$ or simply $u\in\uu$, if (1) $u$ is Lipschitz continuous on $\dd$; (2) $u(0)=0$; (3) $u$ is a stabilizing control, i.e., $\lim_{t\ra\infty}|\phi(t; x_0, u)| = 0$ for all $x_0\in\dd$; and (4) $J(x_0, u)<\infty$.
\end{dfn}
Let $V:\R^n\ra \R$ be the value function for this problem, i.e.,
\begin{equation}\label{E: v}
    V(x):=\inf_{u\in\uu} J(x,u).
\end{equation}
We aim to find $V$ as well as the associated optimal control $u^*$. If we introduce 
\begin{equation}
    G(x,u,p):= \lk(x, u)+p\cdot(f(x)+g(x)u), 
\end{equation}
where $p\in \R^n$, 
and define the Hamiltonian
\begin{equation}
    H(x, p) = \sup_{u\in\R^m}-G(x,u,p), 
\end{equation}
then $V$ is generally a viscosity solution (see Appendix \ref{app: vis} for a formal definition) within $C(\dd)$ to the HJB equation 
\begin{equation}
    H(x,  DV(x)) = 0. 
\end{equation}
We show this in the proof of Proposition \ref{prop: uniqueness_H}. 

\begin{rem}
Basic properties of viscosity solutions are discussed in Appendix \ref{app: vis}. The concept of viscosity solution relaxes $C^1$ solutions. Note that, at differentiable points,  $DV(x)$ exists and $\{DV(x)\}=\partial^+V(x) = \partial^-V(x)$. In this case, to justify a viscosity solution, we can simply substitute $DV(x)$ and check if $F(x, V(x), DV(x))=0$ pointwise. %
If $V$ is not differentiable at a given point, then we have to go through the conditions in Definition \ref{dfn: vis1} to verify that $V$ is a viscosity solution. A classical example is that $V(x)= 1- |x|$, $x\in\R$, is a viscosity solution to $|DV|-1=0$ with boundary conditions $V(-1)=V(1)=0$. To verify this, we only have to check if (1) and (2) in Definition \ref{dfn: vis1} are satisfied at $0$. \qed
\end{rem}

We also provide the following nice properties and complete the proofs in Appendix \ref{app: proof_sec_2}. 
\begin{prop}[Dynamic Programming Principle]\label{prop: DPP}
For all $x\in\R^n$ and $t>0$, 
\begin{equation}
    V(x)=\inf_{u\in\uu}\left\{\int_0^t \lk(\phi(s;x,u), u(s))ds + V(\phi(t;x,u))\right\}. 
\end{equation}
\end{prop}

\begin{prop}[Uniqueness of Viscosity Solution]\label{prop: uniqueness_H}
The $V$ defined in \eqref{E: v} is the unique viscosity
solution of $H(x,DV(x))=0.$
\end{prop}

\begin{thm}[Optimal Feedback Control]\label{thm: optimal_feedback}
Let $\kappa: \Omega\ra \R^m$ be locally Lipschitz continuous. Suppose that $u^*(\cdot):=\kappa(\phi(\cdot))$ and $u^*\in\uu$.  If $V$ is the viscosity solution of 
$G(x, \kappa(x), DV(x))=0$, 
then 
$J(x, \kappa(\phi(\cdot)))=V(x).$
\end{thm}

Note that $G(x, u, p)$ is minimized given that $u(x) = -\frac12 R^{-1} g^T(x)p^T$. With this, the HJB equation reduces to 
\begin{equation}\label{E: HJB_optimal}
\begin{split}
     H(x, DV(x)) 
    = &-Q(x) - DV(x)\cdot f(x) \\
    &+ \frac{1}{4} DV(x)g(x)R^{-1}g^T(x)(DV(x))^T. 
\end{split}
\end{equation}
We can either numerically solve the nonlinear PDE \eqref{E: HJB_optimal} or use policy iteration to approximate $V$ and the optimal controller. The conventional policy iteration \cite{bardi1997optimal, jiang2017robust} assumes that $V\in C^1(\Omega)$ and  seeks $C^1$ solutions $V_i$ to the  GHJBs $G_i(x, u_i, DV_i(x))=0$ for each $i\in\{0,1,\cdots\}$, where $u_i=-\frac{1}{2}R^{-1}g^T(x)DV_i(x)$ for $i\in\{1, 2, \cdots\}$ and $u_0\in\uu$. The convergence value function $V_\infty$ is expected to solve \eqref{E: HJB_optimal} and $u_i\ra u^*$ at least pointwise \citep[Theorem 3.1.4, ][]{jiang2017robust}. The numerical solution of GHJBs, which are linear, is commonly believed to be achieved more easily.

However, the continuous differentiability (on $\Omega$) of  $\{V_i\}$ and $V$ are assumed %
without justification \cite{bardi1997optimal, jiang2017robust}, leading to uncertainty regarding the applicability of the obtained results.  Even though $V_i\in C^1(\Omega)$ for all $i$, the limit $V_\infty$ w.r.t. the uniform norm in \citep[Theorem 3.1.4, ][]{jiang2017robust} may not be continuously differentiable, and hence may not be the approximation of $V\in C^1(\Omega)$. Motivated by this, we characterize solutions to GHJB equations and demonstrate in Section \ref{sec: exact_PI} that the exact policy iteration based on viscosity solutions converges to the viscosity solution of the HJB. The convergence analysis differs from the conventional case. Based on this analysis, we show in Section \ref{sec:algos} that the neural policy iteration algorithms, ELM-PI and PINN-PI, also converge to viscosity solutions under less restrictive assumptions. The algorithms will be presented in Section \ref{sec:algos}, and the convergence analysis will be discussed in Section \ref{sec:convergence}.

\section{Algorithms}
\label{sec:algos}

\subsection{Exact policy iteration}\label{sec: exact_PI}

To begin, we provide an overview of the theoretical foundation of policy iteration. Policy iteration (PI) originates optimal control of Markov decision processes (MDP) \citep{bellman1957dynamic,howard1960dynamic} (additional references can be found in recent texts and monographs \citep{bertsekas2012dynamic,bertsekas2019reinforcement}). In this section, we present a fundamental version of PI for the system (\ref{E:sys}) with the cost (\ref{E:cost}). The algorithm dates back to the 1960s \cite{leake1967construction,milshtein1964successive,vaisbord1963concerning,kleinman1968iterative}, and its convergence is established in various sources \cite{saridis1979approximation, beard1995improving,milshtein1964successive,vaisbord1963concerning,jiang2017robust,farsi2023reinforcement}. 
However, all these proofs rely on strong assumptions on the smoothness of the optimal value function, which may or may not be satisfied in general by the solutions to the HJB equation associated with the optimal control problem. To illustrate this, consider the bilinear scalar problem $\dot x=xu$, with $Q(x)=x^2$ and $R=1$. The optimal value function is $V(x)=2\abs{x}$, which fails to be differentiable at $x=0$. We provide a regularity analysis in the general setting where viscosity solutions are allowed (see Section \ref{sec:convergence}). 

We now define policy iteration with exact solutions to PDEs, which we refer to as exact-PI. This process begins with an initial policy $u=\kappa_0(x)$, where $\kappa_0(\zero)=\zero$. This initial policy is assumed to be an admissible controller. For each $i\ge 0$, exact-PI performs the following two steps iteratively:
\begin{enumerate}
	\item (\textbf{Policy evaluation}) Compute a value function $V_i(x)$ at all $x\in\Omega\setminus\{\zero\}$ for the policy $\kappa_i$ by solving the GHJB
	\begin{equation}\label{eq:pi_evaluate}
	\begin{split}
	 G_i(x, \kappa_i(x), DV_i(x))  := &
	     	Q(x) + \kappa_i^T(x) R(x) \kappa_i(x)\\
       & +  DV_i(x) (f(x)+g(x)\kappa_i(x))\\
	     	 = & 0. 
	\end{split}
	\end{equation}
	We set $V_i(\zero)=0$ for all $i\geq 0$. 
	
	\item (\textbf{Policy improvement}) Update the policy 
	\begin{equation}\label{eq:pi_improve}
	\kappa_{i+1}(x) = \left\{\begin{array}{lr} 
-\frac{1}{2}R^{-1} g^T(x) (DV_i(x))^T, \;\;\text{if}\; x\neq \zero;\\
\zero \;\;\text{otherwise.}
\end{array}\right.  
\end{equation}
\end{enumerate}

The exact-PI algorithm is impractical because exact solution of the linear PDE (\ref{eq:pi_evaluate}) is generally unavailable. In the following sections, we propose two algorithms for neural policy iterations by solving this PDE iteratively using function approximators.  

\subsection{ELM-PI via linear least squares}

The first algorithm uses a one-layer function of the form
\begin{equation}
    \label{eq:Vbasis}
    \hv(x):=V(x;\beta) = \beta^T \sigma(Wx+b),
\end{equation}
where $\beta\in\Real^m$, $W\in\Real^{m\times n}$, $b\in \Real^m$, and $\sigma:\,\Real\ra\Real$ is an activation function applied element-wise. It can be easily verified that the gradient of $V$ with respect to $x$ takes the form
\begin{equation}
    \label{eq:dVbasis}
    DV(x; \beta) = \beta^T \text{diag}(\sigma'(Wx+b))W ,
\end{equation}
where $\text{diag}$ maps the $m$-vector $\sigma'(Wx+b)$ to an $m\times m$ diagonal matrix and $\sigma'(\cdot)$ is the derivative of $\sigma$ applied element-wise. 

We will briefly describe how to solve a PDE via optimization in this paragraph. Suppose we would like to solve a PDE $H(x,DV)=0$. A general idea for solving a PDE via optimization is to collect a number of collocation points $\set{x_s}_{s=1}^N$, on which we evaluate the derivative $DV(x;\beta)$ of a parameterized,  potential solution $V(x;\beta)$ and formulate the residual loss with mean squared error (MSE) as
\begin{equation}
    \label{eq:lossV}
    \begin{split}
         &\text{Loss}(\beta)\\
         =& \frac{1}{N}\sum_{s=1}^N H(x_s,DV(x_s;\beta))^2 + \lambda \sum_{p=1}^{N_b}(V(y_p;\beta)-\hat{V}(y_p))^2, 
    \end{split}
\end{equation}
where $\lambda>0$ is a weight parameter, the points $\set{y_p}_{p=1}^{N_b}$ are boundary points and $\hat{V}(y_p)$ describes the boundary value at these points. 

To efficiently solve (\ref{eq:pi_evaluate}) by (\ref{eq:Vbasis}), the main idea is to randomize $W$ and $b$ and then fix them when optimizing $\text{Loss}(\beta)$ defined by (\ref{eq:lossV}). Due to the linearity of $V(x;\beta)$ in $\beta$, the linearity of the PDE (\ref{eq:pi_evaluate}), and definition of $\text{Loss}(\beta)$ by (\ref{eq:lossV}), we obtain a linear least square optimization problem, which can be solved efficiently and accurately for moderate sized problems. We call this ELM-PI\footnote{We choose the name ELM-PI over LS-PI because we essentially solve the PDE via a similar architecture to an extreme learning machine (ELM) \cite{huang2006extreme} for solving PDEs \cite{dong2021local}. LS-PI in fact exists in the literature for solving control problems with finite state and action spaces \cite{lagoudakis2003least} where PDE solving is irrelevant. Here we want to address the more difficult problem of solving continuous control problems with continuous state and action spaces, where PDE solving seems inevitable if one is to compute the optimal solutions. 
} and describe it in Algorithm \ref{alg:elm-pi}. Here, the boundary condition\footnote{Alternatively, in this setting, we can simply set $V(\zero)=0$ by subtracting a nonzero $V(\zero)$ as a bias term. This does not affect the subsequent controller.} is simply $V(\zero)=0$, because the the origin is an equilibrium point without $u=0$. From (\ref{eq:pi_evaluate}), (\ref{eq:Vbasis}), and (\ref{eq:dVbasis}), the loss function of $\beta$ reduces to 
 \begin{equation}
\label{eq:lsqloss}
\begin{split}
   & \text{Loss}(\beta)\\
    =& \frac{1}{N}\sum_{s=1}^N H(x_s,\beta^T \text{diag}(\sigma'(Wx+b))W)^2 + \lambda (\beta^T \sigma(b))^2, 
\end{split}  
\end{equation}
where $H$ is given by the left-hand side of (\ref{eq:pi_evaluate}) and linear in $DV$. Hence minimizing (\ref{eq:lsqloss}) with $\beta$ is a linear least square problem.

From our experiments, it appears immaterial whether Steps 2 and 3 of Algorithm \ref{alg:elm-pi} are placed within or outside of the loop. In other words, we can use the same set of parameters $W$ and $b$ as well as the set of collocation points for all iterations.

\subsection{PINN-PI via physics-informed neural network}

Physics-informed neural networks (PINN) \cite{raissi2019physics} are a popular method for solving PDEs. We propose a variant for performing physics-informed neural policy iteration in this subsection. 

For each $i$, instead of assuming that a solution $V_i(x)$ to Equation (\ref{eq:pi_evaluate}) takes the form (\ref{eq:Vbasis}), we consider a more general approach by assuming it to be a neural network function:
\begin{equation}
\label{eq:neuralV}
\hat{V}_{i}(x) := V_{i,\text{NN}}(x;\theta),
\end{equation}
where $V_{i, \text{NN}}$ represents a feedforward neural network with potentially multiple layers and nonlinear activation functions. In this formulation, $\theta$ represents the parameters of the neural network, allowing for a flexible and adaptable representation of the solution $V_i(x)$. Even with just one hidden layer, the PINN approach would be allowed to change all parameters in the optimization process, leading to a non-convex optimization problem. Gradient descent methods are usually used to solve these large-scale non-convex optimization problems.  

Similar to ELM-PI, at each iteration, we choose a set of collocation points $\set{x_s}_{s=1}^N$, evaluate the derivatives $DV_{i, \text{NN}}$ of $V_{i, \text{NN}}$ at these points using automatic differentiation, and form a residual loss 
 \begin{equation}
\label{eq:NNloss}
    \text{Loss}(\theta)= \frac{1}{N}\sum_{s=1}^N H(x_s,DV_{i, \text{NN}}(x_s;\theta))^2 + \lambda (V_{i, \text{NN}}(0;\theta))^2, 
\end{equation}
which is in general a non-convex function of $\theta$. 

We describe PINN-PI in Algorithm \ref{alg:pinn-pi}.

\begin{algorithm}[H]
	\caption{Extreme Learning Machine Policy Iteration  (ELM-PI)}\label{alg:elm-pi} 
	\begin{algorithmic}[1]
		\REQUIRE  $f$, $g$, $Q$, $R$, $k_0$, $\Omega$, $N$, $m$
		\REPEAT 
            \STATE Generate random $W$ and $b$ 
            \STATE Generate random $\set{x_s}_{s=1}^N\subset\Omega$
		\STATE Finding $\beta$ that minimizes (\ref{eq:lsqloss}) to form $V_i$ from (\ref{eq:Vbasis})
		\STATE Update $\kappa_{i+1}$ according to (\ref{eq:pi_improve})
		\STATE $i=i+1$
		\UNTIL{desired accuracy max iterations reached}
	\end{algorithmic}
\end{algorithm}

\begin{algorithm}[H]
	\caption{Physics-Informed Neural Network Policy Iteration  (PINN-PI)}\label{alg:pinn-pi} 
	\begin{algorithmic}[1]
		\REQUIRE  $f$, $g$, $Q$, $R$, $k_0$, $\Omega$, $V_{i, \text{NN}}(x;\theta)$
		\REPEAT 
            \STATE Generate random $\set{x_s}_{s=1}^N\subset\Omega$ 
            \REPEAT
                \STATE Run gradient descent on $\theta$ with (\ref{eq:NNloss})
            \UNTIL{desired accuracy or max epochs reached}		
            \STATE Form $V_i(x)$ from $V_{i, \text{NN}}(x;\theta)$
            \STATE Update $\kappa_{i+1}$ according to (\ref{eq:pi_improve})
		\STATE $i=i+1$
		\UNTIL{desired accuracy or max iterations reached}
	\end{algorithmic}
\end{algorithm}

A natural question to ask is when to terminate the algorithm. Clearly, there is no guarantee that gradient descent will find a global minimum $\theta^*$ for (\ref{eq:NNloss}). Even if it does, the resulting $V_{i, \text{NN}}(x;\theta^*)$ will not satisfy (\ref{eq:pi_evaluate}) precisely. What one can hope for is that when the observed loss is sufficiently small and the number of iterations become large, $V_{i, \text{NN}}(x;\theta^*)$, where $\theta^*$ is a returned minimizer, can approximate the optimal solution to an arbitrary precision. In Section \ref{sec:convergence}, we provide a convergence analysis and further discussion on this issue. In practice, the algorithm will terminate when either a desired accuracy is reached or a predetermined number of iterations is completed. In such cases, due to the approximation errors, it is unclear whether the resulting controller is stabilizing. We provide a verification framework to address this issue, which we discuss in Section \ref{sec:verify}. 

\subsection{Loss term to ensure local stability is preserved across iterations}

Based on our observations, training optimal controllers for high-dimensional systems remains \textit{extremely} challenging. In fact, most state-of-the-art reinforcement learning algorithms, whether model-free or model-based, struggle to solve the benchmark control problems we chose with stability guarantees, even in a small region around the equilibrium point to be stabilized (such as cartpole and quadrotors). 

Through our extensive testing (see Table \ref{tab:synthetic} in Appendix \ref{sec:appendix:experiments}), we found that ELM-PI excels in solving low-dimensional problems with high accuracy and fast solver time. However, PINN-PI scales better with state dimensions. Hence, we focus on the PINN-PI algorithm for high-dimensional control problems. 

When naively implemented, PINN-PI can also lead to unstable controllers. This is because the loss function (\ref{eq:NNloss}) does not capture the stabilization requirement of the resulting controller 
\begin{equation}\label{eq:neural_control_update}
  \hat{\kappa}_{i+1}(x) = -\frac{1}{2}R^{-1} g^T(x) (DV_{i, \text{NN}}(x))^T. 
\end{equation}

To overcome this issue, we draw inspiration from classical control theory. When $f(x)=Ax$ and $g(x)=B$, the exact-PI algorithm is nothing but a sequence of Lyapunov equations that can be used to iteratively solve the algebraic Riccati equation. Given the assumptions on $f$ and $g$, locally (\ref{E:sys}) is approximated by $\dot x=Ax+Bu$, where $A=Df(0)$ and $g(0)=B$. 

We examine the linear approximation of $\hat{k}$ and quadratic approximation of $V_{i, \text{NN}}(x)$ around the origin. Assume $\nabla ^2 Q(0)=\hat Q>0$ and let $\hat R = R(0)$. Furthermore, suppose that, for each $i\ge 0$, the controller $\hat k_i$ is exponentially stabilizing and denote $\hat{K}_i:=D\hat{k}_i(0)$. Write $\hat A_i=A+B\hat K_i$. Since $\hat A_i$ is Hurwitz, by linear system theory, there exists a quadratic function $\hat P$ that solves the Lyapunov equation  
\begin{equation}\label{eq:linear_pi}
\hat P_i \hat A_i + \hat A_i^T\hat P_i = - \hat Q - \hat K_i^T \hat R \hat K_i.    
\end{equation}
Comparing this with (\ref{eq:pi_evaluate}), we expect the quadratic part of $V_{i, \text{NN}}(x)$ to be approximated by $x^T \hat P x$ near the origin and the next neural controller $\hat{\kappa}_{i+1}(x)$ is well approximated by a linear controller $\hat K_{i+1}x$ near the origin with 
\begin{equation}\label{eq:predicted_linear_gain}
\hat K_{i+1} = - R^{-1}B^T \hat P_{i},
\end{equation}
which is precisely the gain update required for policy improvement for linear systems. Hence we expect that 
\begin{equation}\label{eq:gain_match}
\hat K_{i+1} =D\hat{\kappa}_{i+1}(0). 
\end{equation}
In view of (\ref{eq:neural_control_update}), this can be easily encoded as a loss term 
\begin{equation}
    \label{eq:loss_gain_match}
    \norm {\frac{\partial}{\partial x}(-\frac{1}{2}R^{-1} g^T(x) (DV_{i, \text{NN}}(x))^T)\big\vert_{x=0} - \hat K_{i+1}}_\text{F} 
\end{equation}
where $\norm{\cdot}_\text{F}$ is the Frobenius norm and $K_{i+1}$ is solved by (\ref{eq:predicted_linear_gain}) and (\ref{eq:linear_pi}). This loss term plays a significant role in stabilizing the training process of PINN-PI for high-dimensional systems.

\subsection{Verification of stability via neural Lyapunov functions}\label{sec:verify}

Upon termination of Algorithms \ref{alg:elm-pi} or \ref{alg:pinn-pi}, we obtain an approximation $\hv(x)$ of the optimal value function. A corresponding approximate optimal control is given by
\begin{equation}
    \label{eq:controlV}
    	u = \hk(x) %
     = \left\{\begin{array}{lr} 
-\frac{1}{2}R^{-1} g^T(x) (D\hv(x))^T, \;\;\text{if}\; x\neq \zero;\\
\zero \;\;\text{otherwise.}
\end{array}\right.   
\end{equation}
For either Algorithms \ref{alg:elm-pi} or \ref{alg:pinn-pi}, $D\hv$ can be readily computed, and is a function involving nonlinear activation functions and possible compositions of them when using multi-layer neural networks in Algorithm \ref{alg:pinn-pi}. 

Suppose that the algorithms terminate perfectly as in the exact-PI case, we have $\hv_{i+1}(x)=\hv_i{(x)}=V(x)$ and $\hk_{i+1}(x)=\hk_{i}(x)=\kappa(x)$. We obtain from (\ref{eq:pi_evaluate}) that 
\begin{equation}
    \label{eq:dV}
    \begin{split}
        &DV(x)(f+g(x)\kappa(x)) \\
        =& -Q(x) + \kappa^T(x) R(x) \kappa(x) < 0,\quad x\neq \zero,
    \end{split} 
\end{equation}
provided that $Q(x)$ is positive definite. However, because of the use of function approximators, we cannot obtain (\ref{eq:dV}). Instead, we use a satisfiability modulo theories (SMT) solver \cite{gao2013dreal} to verify the following nonlinear inequality
\begin{equation}
    \label{eq:dVeps}
    D\hv(x)(f+g(x)\hk(x)) \le - \mu,\quad x\in \Omega \setminus U_\eps,
\end{equation}
where $U_\eps$ is a small neighborhood around the origin of radius $\eps>0$, and $\mu>0$ is a small constant. While checking the exact satisfaction of inequality (\ref{eq:dV}) is in general undecidable, there exist delta-complete SMT solvers \cite{gao2013dreal} that can either verify the inequality or falsify a $\delta$-weakened version of it, where $\delta>0$ can be any arbitrary precision parameter. To use such tools, it is necessary to exclude a small neighborhood of the origin \cite{chang2019neural,zhou2022neural}, as  (\ref{eq:dV}) turns into an equation at the origin. It is worth noting that, with additional assumptions, one may be able to verify exact stability including the origin through examination of the derivatives of the vector fields \cite{liu2023towards}. In this paper, we verify the stability given the controllers generated from Algorithms \ref{alg:elm-pi} and \ref{alg:pinn-pi} using (\ref{eq:dVeps}) for a $U_\eps$ with %
some small $\eps>0$, which ensures that solutions are attracted to any prescribed small neighborhood of the origin. Note that, by the continuity (or smoothness) of the approximators, for sufficiently small $\eps>0$, we can also have 
\begin{equation}\label{E: dV_nbhd}
D\hv(x)(f+g(x)\hk(x)) \le - \mu +\mathcal{O}(\eps)<0,\quad x\in U_\eps \setminus\{\zero\} ,
\end{equation}
where $\mathcal{O}(\eps)\ra 0$ as $\eps\ra0$. Assuming that a suitable value for $\eps$ can be chosen, such that both \eqref{eq:dVeps} and \eqref{E: dV_nbhd} hold, the stability can be verified using neural Lyapunov functions.

\section{Convergence analysis}
\label{sec:convergence}

We state the main regularity and convergence results in this section. The proofs can be found in Appendix \ref{app: proof_sec_4}. 
\subsection{Convergence analysis for exact-PI}\label{sec: conv_exact}
In view of Proposition \ref{prop: uniqueness_H}, we expect that each policy evaluation in exact-PI has a unique solution so that the algorithm eventually yields a meaningful outcome. We first establish that each GHJB in exact-PI possesses a unique viscosity solution characterized by a specific pattern.

\begin{prop}\label{prop: uniqueness_GHJB}
Let $u\in\uu$ be any (autonomous) state feedback controller so that there exists some feedback policy $\kappa: \R^n\ra\R$ such that $u(\cdot)=\kappa(\phi(\cdot))$. Then the infinitesimal dynamic $-G(x, u, DV(x))=0$ has a unique positive definite viscosity solution within the space $C(\Omega)\cap C^1(\Omega\setminus\{\zero\})$.
\end{prop}

\begin{cor}\label{cor: iteration}
For each $i\geq 0$, the GHJB $G_i(x, \kappa_i(x), DV_i(x))=0$ has a unique positive definite viscosity solution $V_i$,  which belongs to $C(\Omega)\cap C^1(\Omega\setminus\{\zero\})$.
\end{cor}

\begin{rem}\label{rem: zeros}
Note that in the situation where the state feedback controller $u$ is not necessarily stabilizing, but (1)(2)(4) of Definition \ref{dfn: adm} still hold and $f(x)+g(x)u$ has countable zeros, the above existence and uniqueness of viscosity solution still follow. However, this discussion is beyond the scope of this paper. \qed
\end{rem}

The following theorem states that exact-PI converges to the true solution to \eqref{E: HJB_optimal}. 

\begin{thm}[Convergence of Successive Approximations of Viscosity Solution]\label{thm: convergence}
For each $i\geq 0$, let $u_i(\cdot)=\kappa_i(\phi(\cdot))$, where $\kappa_i$ is defined in \eqref{eq:pi_improve}.
Suppose that $u_0\in\uu$, then,
\begin{itemize}
\item[(1)] $u_i\in\uu$ for all $i\in\{0, 1, \cdots\}$. 
    \item[(2)] $V^*\leq V_{i+1}\leq V_i$ for all $x\in\Omega$ and for all $i\in\{0,1,\cdots\}$, where $V_i$ is the viscosity solution to $G_i(x, \kappa_i(x), DV_i(x))=0$ and $V^*$ is the viscosity solution to \eqref{E: HJB_optimal}. 
    \item[(3)] $V_i\ra V^*$ uniformly on $\Omega$ as $i\ra\infty$ given the compactness of $\Omega$. 
\end{itemize}
\end{thm}

\subsection{Convergence analysis for policy iteration using neural approximations}\label{sec: conv_PINN}

The main idea is to formalize properties of the loss function that capture the desired convergence of neural approximations to true solutions, which in this context are the viscosity solutions to the GHJB and HJB equations. We expect that when the training error (or the loss function in \eqref{eq:lsqloss}  or 
\eqref{eq:NNloss}) is small,  the generalization error is also small. In practice, this requires that the number of collocation points chosen from $\Omega$, at which the residual of each GHJB $G_i$ is evaluated, be sufficiently large.  

However, based on Corollary \ref{cor: iteration}, the viscosity solution for each iteration does not exhibit uniform differentiability across the entire domain of $\Omega$. In addition, most convergence results for data-driven methods are typically based on a compact subset of the state space. %
Therefore, direct consideration of $C^1$-uniform convergence on $\Omega \setminus\{\zero\}$ is not feasible. Instead, we achieve the $C^1$-uniform convergence on $\Omega\setminus U_\eps$ and a weaker (asymptotic)  convergence on $U_\eps$, where $U_\eps$ is some open set centered at $\zero$ of arbitrarily small radius $\eps>0$.

To circumvent complex notation, let us consider the general case for any GHJB $G(x,\kappa(x), DV(x))=0$ with admissible $\kappa$ as in Proposition \ref{prop: uniqueness_GHJB} to illustrate the idea.
Focusing on $\Omega\setminus U_\eps$, we consider the space of continuously differentiable functions  $\mathcal{G}=C^1(\Omega\setminus U_\eps,\Real)$  equipped with the $C^1$-uniform norm, $\norm{V}_{C^1} := \sup_{x\in \Omega\setminus U_\eps} \norm{V(x)} + \sup_{x\in \Omega\setminus U_\eps}\norm{DV(x)}$. We consider a training error  
$E_{T, N}:\, \mathcal{G}\to [0,\infty)$ of the following form (e.g. the loss function in \eqref{eq:lsqloss} and \eqref{eq:NNloss}), 
 \begin{equation*}
     \begin{split}
         E_{T,N}(V) = &\frac{1}{N} \sum_{k=1}^N{\norm{G(x_k, \kappa(x_k), DV(x_k))}^2}
 		+\norm{V(\zero)}^2, 
     \end{split}
 \end{equation*}
where $N\in\N$ is associated with the number of collocation points chosen from $\Omega$.  We seek approximations $\{\hv_{ N} \}_{N\in \N}\subseteq \mathcal F$ of the unique viscosity solution $V$ in some function space $\mathcal{F}\subseteq \mathcal{G}$, for instance, the space of functions representable by a one hidden-layer network. Then, we aim to determine whether   $E_{T,N}(\hv_{ N})\ra 0$ implies $\hv_{ N} \ra V\text{ in } \mathcal G$.

Continuing the above settings, in the following proposition, %
we state that by incorporating additional assumptions, a convergence result %
can be obtained on $\Omega\setminus U_\eps$.
\begin{hyp}\label{hyp: lip}
For any Lipschitz continuous function $h$ and its smooth neural approximations $\{\hat{h}_N\}_{N\in\N}$, the Lipschitz constant of $\hat{h}_N$ converges to the true Lipschitz constant as  $\frac{1}{N} \sum_{k=1}^N{\norm{\hat{h}(x_k))}^2}$ converges to $0$. 
\end{hyp}

\begin{rem}
This phenomenon has been thoroughly investigated by \cite{khromov2023some}. For low-dimensional systems, it is possible to also directly penalize the Lipschitz constant of the residual and achieve higher accuracy (see the proof of Proposition \ref{cor: pinn} for details). In contrast, when the Lipschitz constant of the residual is difficult to verify, it is reasonable to assume Hypothesis \ref{hyp: lip}. \qed
\end{rem}

\begin{prop}\label{cor: pinn}
Let $\eps>0$ any arbitrarily small number and  $U_\eps$ be an open set centered at $\zero$ of radius $\eps$.
Let $\mathcal F \subseteq \mathcal{G}$ be a subspace with uniformly bounded Lipschitz constant on $\Omega\setminus U_\eps$. Suppose that $\set{x_k}_{k\in\N} $ is a sequence dense on $\Omega\setminus U_\eps$ with the additional requirement that, for all $N\in \N$,
 		$\delta_N= \inf\set{\delta >0: \Omega\setminus U_\eps\subseteq \bigcup_{k=1}^N\B[\Omega\setminus U_\eps]{\delta}{x_k}}$$ \text{ and } $$ C=\sup\set{ N\mu(\B[\Omega\setminus U_\eps]{\delta_n}{x_1}) : N\in \N} <\infty
 		$
		where $\mu$ is the Lebesgue measure and $\B[\Omega\setminus U_\eps]{\delta}{x}$ is the open ball of radius $\delta >0$ centered at $x$ in $\Omega\setminus U_\eps$. 

Suppose that Hypothesis \ref{hyp: lip} holds and the training error $E_{T,N}(\hv_N)$ can be arbitrarily small for sufficiently large $N$. 
Then,  the neural network $\hv_N \ra V$ in $\mathcal{G}$. 

\end{prop}

\begin{rem}\label{rem: pinn_conv}
The additional requirement on  $\{x_k\}_{k\in\N}$ indicates that the smallest volume of the ``finite coverings'' on $\Omega\setminus U_\eps$ is finite. %
As a technical matter, instead of using the dense set $\{x_k\}$, we can use a sequence of finite sets $A_N:=\{x_k\}_{k=1}^N$ (as in Algorithm \ref{alg:pinn-pi}) that is ``eventually dense" in $\R^n$, i.e. $A_N\to X$ in the Hausdorff metric. %
This would allow one to use new training points,  as long as the finite sets are good approximations of  $\R^n$. \qed
\end{rem}
By patching the above sound approximation on $\Omega\setminus U_\eps$ for any small $\eps>0$ and the asymptotic approximation within $U_\eps$, one can obtain the following convergence guarantee. 
\begin{thm}\label{thm: conv}
    Given $\kappa_0$, let $\{V_i\}$ and $\{\kappa_{i+1}\}$ be updated by exact-PI. 
Let $\{\hv_i\}$ and $\{\hk_{i+1}\}$ be updated by PINN-PI with $\hk_0=\kappa_0$. Let the conditions in Proposition \ref{cor: pinn} be held. Then, for any $i\geq 0$ and $\tta>0$, we can choose a sufficiently dense set of collocation points $\{x_k\}_{k=1}^N$ such that
$$|\hv_i(x)-V_i(x)|\leq \tta, \;|\hk_{i+1}(x)-\kappa_{i+1}(x)|\leq \tta, \;\;x\in\Omega.$$
\end{thm}

\section{Numerical experiments}
\label{sec:examples}

In this section, we present numerical examples to evaluate the performance of the proposed algorithms. We aim to accomplish three goals: 1) Evaluate the performance characteristics of ELM-PI and PINN-PI, ranging from low to high-dimensional systems; 2) Compare with approaches in classical control literature and demonstrate the superior performance of ELM-PI in solving low-dimensional systems and highlight the importance of formal verification; 3) Demonstrate the superior capabilities of PINN-PI in solving high-dimensional benchmark control problems and compare it with state-of-the-art model-free and model-based reinforcement learning algorithms.

\subsection{Synthetic $n$-dimensional nonlinear control}

Consider the nonlinear control problem given by $f_i(x)=x_i^3 + u_i$, where $i=1,2,\ldots,n$. The cost is defined by $Q(x)=\sum_{i=1}^n (x_i^2 + 2x_i^4)$ and $R=I_n$. From optimal control theory, the optimal value function is obtained by solving the HJB equation, giving 
$
V^*(x)= \sum_{i=1}^n \left(\frac12 x_i^4 + \frac12  (x_i^2 + 1)^2 - \frac12\right).
$ 
We run ELM-PI and PINN-PI and compare their performance for various dimensions $n$ in terms of computational time and maximum testing error relative to the true optimal value function. The results are summarized in Table \ref{tab:synthetic} in Appendix \ref{sec:appendix:experiments}. 

From the experimental results, we see that for low-dimensional problems ($n\le 3$), ELM-PI outperforms PINN-PI in terms of both computational efficiency and approximation accuracy. As the dimension increases, more computational units ($m$) are required for ELM-PI to achieve a higher accuracy. For example when $n=4$, to achieve $10^{-3}$ accuracy, it requires $m=3200$, and for $10^{-5}$, $m=6400$. The complexity of solving linear least square is $O(mN^2)$, provided that $N>m$. In our experiments, we set $N=d*m$. Hence, the time complexity is $O(d^2m^3)$, which roughly captures the increase in computational time reported in Table \ref{tab:synthetic} as $m$ and $d$ increase.  For $n\ge 5$, it is evident that ELM-PI becomes inefficient. In comparison, PINN-PI can achieve $10^{-2}$ to $10^{-3}$ accuracy across all dimensions within a reasonable amount of computational time. In fact, we were able to obtain $10^{-2}$  error with $m=800$ across all dimensions. The accuracy, however, does not seem to improve significantly as $m$, $N$, or the number of steps increase in training PINN-PI. 

Based on these evaluations, we recommend to use ELM-PI for low-dimensional problems and PINN-PI for high-dimensional problems. 

\subsection{Inverted pendulum and comparison with successive Galerkin approximations}

We run both ELM-PI and PINN-PI to compute the optimal control and policy for the inverted pendulum (see Section \ref{sec:sga} in the Appendix for more details) Figure \ref{fig:pendulum} in Appendix \ref{sec:appendix:experiments} displays the results of implementing ELM-PI on an inverted pendulum. The value functions, projected onto $x_1$, are plotted for each iteration. The left panel represents the scenario with $m=100$ and $N=200$, while the right panel corresponds to $m=50$ and $N=100$. Despite the visual similarity and apparent convergence after five iterations, it is surprising that the controller obtained from $m=50$ does not actually stabilize the system. Conversely, we have verified that the controller derived from $m=100$ is indeed stabilizing using dReal \cite{gao2013dreal}. This example demonstrates the justification for employing formal verification alongside policy iteration to attain both optimality and stability, particularly in safety-critical scenarios.

We also compare ELM-PI and PINN-PI with successive Galerkin approximations for solving GHJB \cite{beard1997galerkin} and provide further verification results. The conclusion is that while successive Galerkin approximations (SGA) can effectively solve low-dimensional problems, ELM-PI is significantly superior in terms of solver time. Furthermore, as demonstrated in the next section, PINN-PI can solve high-dimensional problems that are beyond the reach of SGA. Due to space limitations, detailed results are included in the supplementary material.

\subsection{Comparison with reinforcement learning algorithms}

We compare PINN-PI against well-established reinforcement learning (RL) algorithms, including Proximal Policy Optimization (PPO) \cite{schulman2017proximal}, Hamilton Jacobi Bellman PPO (HJBPPO) \cite{mukherjee2023bridging}, and Continuous Time Model-Based Reinforcement Learning (CT-MBRL) \cite{yildiz2021continuous}. We train each algorithm in benchmark control environments (inverted pendulum, cartpole, 2D quadrotor, and 3D quadrotor) and compare their control costs.

PPO is a model-free actor-critic algorithm that uses a clipped objective function to limit policy updates, ensuring incremental learning steps. It consists of an actor-network $\pi_\theta(a|s)$ that takes the state $s$ as input and outputs a distribution over actions $a$, and a value network $V_\phi(s)$ that takes the state $s$ as input and outputs the expected return. HJBPPO is an extension of PPO that uses the continuous-time HJB equation as a loss function, instead of the discrete-time Bellman optimality equation. CT-MBRL introduces a model-based approach, employing continuous-time dynamics for more precise control and prediction in RL tasks. While we are aware there are other RL algorithms available in the literature, the rationale for choosing these RL algorithms for comparison is given in Section \ref{sec:rationale}. 

The comparison is presented in Appendix \ref{app:comparison}. While the RL algorithms exhibit comparable performance to PI-Policy in the inverted pendulum environment, it is evident that for higher-dimensional environments like the cartpole and quadrotor, the RL algorithms struggle to achieve stability, leading to their accumulated control cost diverging. PINN-PI, on the other hand, demonstrates convergence to equilibrium in less than two seconds, marking a significant improvement. This may appear striking at first, but being able to encode local asymptotic stability as a loss (\ref{eq:loss_gain_match}) in training plays an important role in achieving an asymptotically stabilizing optimal controller, whereas other RL algorithms typically use episodic training over a finite time horizon. This discrepancy explains the superior performance of PINN-PI in problems where asymptotic stability is closely tied to the performance criteria (think of LQR as a special case).

\section{Conclusions}

We propose two algorithms for conducting model-based policy iterations to solve nonlinear optimal control problems. The first algorithm leverages the linearity of the PDE that defines the policy value and utilizes linear least squares to obtain the approximation. This approach proves to be highly efficient and accurate for low-dimensional problems. The second approach employs physics-informed neural networks and demonstrates better scalability for high-dimensional problems. We emphasize the importance of incorporating formal verification on top of policy iterations to achieve both optimality and stability when safety is a concern. We provide theoretical analysis that shows policy iterations (both exact and approximate) converge to the true optimal solutions in general settings. Limitations of this work and potential future work are discussed in Section \ref{sec:future}.

\section{Impact Statements}

This paper presents work whose goal is to advance the field of Machine Learning. There are many potential societal consequences of our work, none which we feel must be specifically highlighted here. 

\bibliography{ref}

\begin{thebibliography}{46}
\providecommand{\natexlab}[1]{#1}
\providecommand{\url}[1]{\texttt{#1}}
\expandafter\ifx\csname urlstyle\endcsname\relax
  \providecommand{\doi}[1]{doi: #1}\else
  \providecommand{\doi}{doi: \begingroup \urlstyle{rm}\Url}\fi

\bibitem[Bardi et~al.(1997)Bardi, Dolcetta, et~al.]{bardi1997optimal}
Bardi, M., Dolcetta, I.~C., et~al.
\newblock \emph{Optimal control and viscosity solutions of
  Hamilton-Jacobi-Bellman equations}, volume~12.
\newblock Springer, 1997.

\bibitem[Beard(1995)]{beard1995improving}
Beard, R.~W.
\newblock \emph{Improving the closed-loop performance of nonlinear systems}.
\newblock PhD thesis, Rensselaer Polytechnic Institute, 1995.

\bibitem[Beard et~al.(1997)Beard, Saridis, and Wen]{beard1997galerkin}
Beard, R.~W., Saridis, G.~N., and Wen, J.~T.
\newblock Galerkin approximations of the generalized hamilton-jacobi-bellman
  equation.
\newblock \emph{Automatica}, 33\penalty0 (12):\penalty0 2159--2177, 1997.

\bibitem[Beard et~al.(1998)Beard, Saridis, and Wen]{beard1998approximate}
Beard, R.~W., Saridis, G.~N., and Wen, J.~T.
\newblock Approximate solutions to the time-invariant hamilton--jacobi--bellman
  equation.
\newblock \emph{Journal of Optimization theory and Applications}, 96:\penalty0
  589--626, 1998.

\bibitem[Bellman(1957)]{bellman1957dynamic}
Bellman, R.~E.
\newblock \emph{Dynamic Programming}.
\newblock Princeton University Press, 1957.

\bibitem[Bertsekas(2012)]{bertsekas2012dynamic}
Bertsekas, D.~P.
\newblock \emph{Dynamic Programming and Optimal Control: Volume I}, volume~1.
\newblock Athena Scientific, 2012.

\bibitem[Bertsekas(2015)]{bertsekas2015value}
Bertsekas, D.~P.
\newblock Value and policy iterations in optimal control and adaptive dynamic
  programming.
\newblock \emph{IEEE transactions on neural networks and learning systems},
  28\penalty0 (3):\penalty0 500--509, 2015.

\bibitem[Bertsekas(2019)]{bertsekas2019reinforcement}
Bertsekas, D.~P.
\newblock \emph{Reinforcement Learning and Optimal Control}.
\newblock Athena Scientific, 2019.

\bibitem[Bhasin et~al.(2013)Bhasin, Kamalapurkar, Johnson, Vamvoudakis, Lewis,
  and Dixon]{bhasin2013novel}
Bhasin, S., Kamalapurkar, R., Johnson, M., Vamvoudakis, K.~G., Lewis, F.~L.,
  and Dixon, W.~E.
\newblock A novel actor--critic--identifier architecture for approximate
  optimal control of uncertain nonlinear systems.
\newblock \emph{Automatica}, 49\penalty0 (1):\penalty0 82--92, 2013.

\bibitem[Camilli et~al.(2001)Camilli, Gr{\"u}ne, and
  Wirth]{camilli2001generalization}
Camilli, F., Gr{\"u}ne, L., and Wirth, F.
\newblock A generalization of zubov's method to perturbed systems.
\newblock \emph{SIAM Journal on Control and Optimization}, 40\penalty0
  (2):\penalty0 496--515, 2001.

\bibitem[Chang et~al.(2019)Chang, Roohi, and Gao]{chang2019neural}
Chang, Y.-C., Roohi, N., and Gao, S.
\newblock Neural lyapunov control.
\newblock \emph{Advances in Neural Information Processing Systems}, 32, 2019.

\bibitem[Chen et~al.(2022)Chen, Chi, Yang, et~al.]{chen2022bridging}
Chen, J., Chi, X., Yang, Z., et~al.
\newblock Bridging traditional and machine learning-based algorithms for
  solving pdes: The random feature method.
\newblock \emph{arXiv preprint arXiv:2207.13380}, 2022.

\bibitem[Crandall et~al.(1984)Crandall, Evans, and Lions]{crandall1984some}
Crandall, M.~G., Evans, L.~C., and Lions, P.-L.
\newblock Some properties of viscosity solutions of hamilton-jacobi equations.
\newblock \emph{Transactions of the American Mathematical Society},
  282\penalty0 (2):\penalty0 487--502, 1984.

\bibitem[Dong \& Li(2021)Dong and Li]{dong2021local}
Dong, S. and Li, Z.
\newblock Local extreme learning machines and domain decomposition for solving
  linear and nonlinear partial differential equations.
\newblock \emph{Computer Methods in Applied Mechanics and Engineering},
  387:\penalty0 114129, 2021.

\bibitem[Duan et~al.(2016)Duan, Chen, Houthooft, Schulman, and
  Abbeel]{duan2016benchmarking}
Duan, Y., Chen, X., Houthooft, R., Schulman, J., and Abbeel, P.
\newblock Benchmarking deep reinforcement learning for continuous control.
\newblock In \emph{International conference on machine learning}, pp.\
  1329--1338. PMLR, 2016.

\bibitem[Evans(2010)]{evans2010partial}
Evans, L.~C.
\newblock \emph{Partial Differential Equations}, volume~19.
\newblock American Mathematical Society, 2010.

\bibitem[Farsi \& Liu(2023)Farsi and Liu]{farsi2023reinforcement}
Farsi, M. and Liu, J.
\newblock \emph{Reinforcement Learning}.
\newblock Wiley-IEEE Press, 2023.

\bibitem[Gao et~al.(2013)Gao, Kong, and Clarke]{gao2013dreal}
Gao, S., Kong, S., and Clarke, E.~M.
\newblock dreal: An smt solver for nonlinear theories over the reals.
\newblock In \emph{Proceedings of 24th International Conference on Automated
  Deduction}, pp.\  208--214. Springer, 2013.

\bibitem[Han et~al.(2018)Han, Jentzen, and E]{han2018solving}
Han, J., Jentzen, A., and E, W.
\newblock Solving high-dimensional partial differential equations using deep
  learning.
\newblock \emph{Proceedings of the National Academy of Sciences}, 115\penalty0
  (34):\penalty0 8505--8510, 2018.

\bibitem[Howard(1960)]{howard1960dynamic}
Howard, R.~A.
\newblock \emph{Dynamic Programming and Markov Process}.
\newblock MIT Press, 1960.

\bibitem[Huang et~al.(2006)Huang, Zhu, and Siew]{huang2006extreme}
Huang, G.-B., Zhu, Q.-Y., and Siew, C.-K.
\newblock Extreme learning machine: theory and applications.
\newblock \emph{Neurocomputing}, 70\penalty0 (1-3):\penalty0 489--501, 2006.

\bibitem[Jiang \& Jiang(2012)Jiang and Jiang]{jiang2012computational}
Jiang, Y. and Jiang, Z.-P.
\newblock Computational adaptive optimal control for continuous-time linear
  systems with completely unknown dynamics.
\newblock \emph{Automatica}, 48\penalty0 (10):\penalty0 2699--2704, 2012.

\bibitem[Jiang \& Jiang(2014)Jiang and Jiang]{jiang2014robust}
Jiang, Y. and Jiang, Z.-P.
\newblock Robust adaptive dynamic programming and feedback stabilization of
  nonlinear systems.
\newblock \emph{IEEE Transactions on Neural Networks and Learning Systems},
  25\penalty0 (5):\penalty0 882--893, 2014.

\bibitem[Jiang \& Jiang(2017)Jiang and Jiang]{jiang2017robust}
Jiang, Y. and Jiang, Z.-P.
\newblock \emph{Robust Adaptive Dynamic Programming}.
\newblock John Wiley \& Sons, 2017.

\bibitem[Khromov \& Singh(2023)Khromov and Singh]{khromov2023some}
Khromov, G. and Singh, S.~P.
\newblock Some fundamental aspects about lipschitz continuity of neural network
  functions.
\newblock \emph{arXiv preprint arXiv:2302.10886}, 2023.

\bibitem[Kleinman(1968)]{kleinman1968iterative}
Kleinman, D.
\newblock On an iterative technique for riccati equation computations.
\newblock \emph{IEEE Transactions on Automatic Control}, 13\penalty0
  (1):\penalty0 114--115, 1968.

\bibitem[Lagoudakis \& Parr(2003)Lagoudakis and Parr]{lagoudakis2003least}
Lagoudakis, M.~G. and Parr, R.
\newblock Least-squares policy iteration.
\newblock \emph{The Journal of Machine Learning Research}, 4:\penalty0
  1107--1149, 2003.

\bibitem[Leake \& Liu(1967)Leake and Liu]{leake1967construction}
Leake, R. and Liu, R.-W.
\newblock Construction of suboptimal control sequences.
\newblock \emph{SIAM Journal on Control}, 5\penalty0 (1):\penalty0 54--63,
  1967.

\bibitem[Liu et~al.(2023{\natexlab{a}})Liu, Meng, Fitzsimmons, and
  Zhou]{liu2023physics}
Liu, J., Meng, Y., Fitzsimmons, M., and Zhou, R.
\newblock Physics-informed neural network lyapunov functions: Pde
  characterization, learning, and verification.
\newblock \emph{arXiv preprint arXiv:2312.09131}, 2023{\natexlab{a}}.

\bibitem[Liu et~al.(2023{\natexlab{b}})Liu, Meng, Fitzsimmons, and
  Zhou]{liu2023towards}
Liu, J., Meng, Y., Fitzsimmons, M., and Zhou, R.
\newblock Towards learning and verifying maximal neural lyapunov functions.
\newblock \emph{arXiv preprint arXiv:2304.07215}, 2023{\natexlab{b}}.

\bibitem[Meng et~al.(2023)Meng, Zhou, and Liu]{meng2023learning}
Meng, Y., Zhou, R., and Liu, J.
\newblock Learning regions of attraction in unknown dynamical systems via
  zubov-koopman lifting: Regularities and convergence.
\newblock \emph{arXiv preprint arXiv:2311.15119}, 2023.

\bibitem[Milshtein(1964)]{milshtein1964successive}
Milshtein, G.~N.
\newblock On an iterative technique for riccati equation computations.
\newblock \emph{Automation and Remote Control}, 25\penalty0 (3):\penalty0
  298--306, 1964.

\bibitem[Mukherjee \& Liu(2023)Mukherjee and Liu]{mukherjee2023bridging}
Mukherjee, A. and Liu, J.
\newblock Bridging physics-informed neural networks with reinforcement
  learning: Hamilton-jacobi-bellman proximal policy optimization (hjbppo).
\newblock \emph{ICML Workshop on New Frontiers in Learning, Control, and
  Dynamical Systems}, 2023.

\bibitem[Ouyang et~al.(2022)Ouyang, Wu, Jiang, Almeida, Wainwright, Mishkin,
  Zhang, Agarwal, Slama, Ray, et~al.]{ouyang2022training}
Ouyang, L., Wu, J., Jiang, X., Almeida, D., Wainwright, C., Mishkin, P., Zhang,
  C., Agarwal, S., Slama, K., Ray, A., et~al.
\newblock Training language models to follow instructions with human feedback.
\newblock \emph{Advances in Neural Information Processing Systems},
  35:\penalty0 27730--27744, 2022.

\bibitem[Poggio et~al.(2017)Poggio, Mhaskar, Rosasco, Miranda, and
  Liao]{poggio2017and}
Poggio, T., Mhaskar, H., Rosasco, L., Miranda, B., and Liao, Q.
\newblock Why and when can deep-but not shallow-networks avoid the curse of
  dimensionality: a review.
\newblock \emph{International Journal of Automation and Computing}, 14\penalty0
  (5):\penalty0 503--519, 2017.

\bibitem[Raffin et~al.(2021)Raffin, Hill, Gleave, Kanervisto, Ernestus, and
  Dormann]{stable-baselines3}
Raffin, A., Hill, A., Gleave, A., Kanervisto, A., Ernestus, M., and Dormann, N.
\newblock Stable-baselines3: Reliable reinforcement learning implementations.
\newblock \emph{Journal of Machine Learning Research}, 22\penalty0
  (268):\penalty0 1--8, 2021.
\newblock URL \url{http://jmlr.org/papers/v22/20-1364.html}.

\bibitem[Raissi et~al.(2019)Raissi, Perdikaris, and
  Karniadakis]{raissi2019physics}
Raissi, M., Perdikaris, P., and Karniadakis, G.~E.
\newblock Physics-informed neural networks: A deep learning framework for
  solving forward and inverse problems involving nonlinear partial differential
  equations.
\newblock \emph{Journal of Computational physics}, 378:\penalty0 686--707,
  2019.

\bibitem[Saridis \& Lee(1979)Saridis and Lee]{saridis1979approximation}
Saridis, G.~N. and Lee, C.-S.~G.
\newblock An approximation theory of optimal control for trainable
  manipulators.
\newblock \emph{IEEE Transactions on Systems, Man, and Cybernetics}, 9\penalty0
  (3):\penalty0 152--159, 1979.

\bibitem[Schulman et~al.(2017)Schulman, Wolski, Dhariwal, Radford, and
  Klimov]{schulman2017proximal}
Schulman, J., Wolski, F., Dhariwal, P., Radford, A., and Klimov, O.
\newblock Proximal policy optimization algorithms.
\newblock \emph{arXiv preprint arXiv:1707.06347}, 2017.

\bibitem[Silver et~al.(2017)Silver, Schrittwieser, Simonyan, Antonoglou, Huang,
  Guez, Hubert, Baker, Lai, Bolton, et~al.]{silver2017mastering}
Silver, D., Schrittwieser, J., Simonyan, K., Antonoglou, I., Huang, A., Guez,
  A., Hubert, T., Baker, L., Lai, M., Bolton, A., et~al.
\newblock Mastering the game of go without human knowledge.
\newblock \emph{Nature}, 550\penalty0 (7676):\penalty0 354--359, 2017.

\bibitem[Sirignano \& Spiliopoulos(2018)Sirignano and
  Spiliopoulos]{sirignano2018dgm}
Sirignano, J. and Spiliopoulos, K.
\newblock Dgm: A deep learning algorithm for solving partial differential
  equations.
\newblock \emph{Journal of Computational Physics}, 375:\penalty0 1339--1364,
  2018.

\bibitem[Vaisbord(1963)]{vaisbord1963concerning}
Vaisbord, E.~M.
\newblock Concerning an approximate method for optimum control synthesis.
\newblock \emph{Avtomatika i Telemekhanika}, 24\penalty0 (12):\penalty0
  1626--1632, 1963.

\bibitem[Vrabie \& Lewis(2009)Vrabie and Lewis]{vrabie2009neural}
Vrabie, D. and Lewis, F.
\newblock Neural network approach to continuous-time direct adaptive optimal
  control for partially unknown nonlinear systems.
\newblock \emph{Neural Networks}, 22\penalty0 (3):\penalty0 237--246, 2009.

\bibitem[Weinan et~al.(2021)Weinan, Han, and Jentzen]{weinan2021algorithms}
Weinan, E., Han, J., and Jentzen, A.
\newblock Algorithms for solving high dimensional pdes: from nonlinear monte
  carlo to machine learning.
\newblock \emph{Nonlinearity}, 35\penalty0 (1):\penalty0 278, 2021.

\bibitem[Yildiz et~al.(2021)Yildiz, Heinonen, and
  L{\"a}hdesm{\"a}ki]{yildiz2021continuous}
Yildiz, C., Heinonen, M., and L{\"a}hdesm{\"a}ki, H.
\newblock Continuous-time model-based reinforcement learning.
\newblock In \emph{International Conference on Machine Learning}, pp.\
  12009--12018. PMLR, 2021.

\bibitem[Zhou et~al.(2022)Zhou, Quartz, De~Sterck, and Liu]{zhou2022neural}
Zhou, R., Quartz, T., De~Sterck, H., and Liu, J.
\newblock Neural lyapunov control of unknown nonlinear systems with stability
  guarantees.
\newblock \emph{Advances in Neural Information Processing Systems}, 2022.

\end{thebibliography}
\bibliographystyle{icml2024}

\newpage
\appendix
\onecolumn

\section{Basic properties of viscosity solutions}\label{app: vis} 

The definition of viscosity solutions is given below. 

\begin{dfn}\label{dfn: vis1}
Define the superdifferential and the subdifferential sets of $V$ at $x$ respectively as
\begin{subequations}
\begin{align}
& \partial^+V(x) \nonumber\\
=&\left\{p\in \R^n: \;\limsup_{y\ra x}\frac{V(y)-V(x)-p\cdot (y-x)}{|y-x|}\leq 0\right\}\label{E: compare1a},\\
& \partial^-V(x) \nonumber\\
=&\left\{q\in \R^n: \;\liminf_{y\ra x}\frac{V(y)-V(x)-q\cdot (y-x)}{|y-x|}\geq 0\right\}\label{E: compare1b}.
\end{align}
\end{subequations}
A continuous function $V$ of a PDE of the form $F(x,V(x),DV(x))=0$ (possibly encoded with boundary conditions) is a viscosity solution if the following conditions are satisfied:
\begin{itemize}
    \item[(1)] (viscosity subsolution) $F(x, V(x), p)\leq 0$ for all $x\in\R^n$ and for all $p\in \partial^+V(x)$.
    \item[(2)] (viscosity supersolution) $F(x, V(x), q)\geq 0$ for all $x\in\R^n$ and for all $q\in \partial^-V(x)$.
\end{itemize}
\end{dfn}

The following lemma \citep[Lemma 1.7, Lemma 1.8, Chapter I, ][]{bardi1997optimal} provides some insights on $\partial^+V(x)$ and $\partial^+V(x)$ for some $V\in C(\dd)$. 

\begin{lem}[Sub- and Supperdifferential]\label{lem: facts}
Let $V\in C(\dd)$. Then
\begin{itemize}
    \item[(1)] $p\in \partial^+V(x)$ if and only if there exists $\psi\in C^1(\dd)$ such that $D\psi(x)=p$ and $u-\psi$ has a local maximum at $x$;
    \item[(2)] $q\in \partial^-V(x)$ if and only if there exists $\psi\in C^1(\dd)$ such that $D\psi(x)=q$ and $u-\psi$ has a local minimum at $x$;
    \item[(3)] if for some $x$ both $\partial^+V(x)$ and  $\partial^-V(x)$ are nonempty, then $\partial^+V(x)=\partial^-V(x)=\{DV(x)\}$;
    \item[(4)] the sets  $\{x\in\dd: \partial^+V(x)\neq \emptyset\}$ and $\{x\in\dd: \partial^-V(x)\neq \emptyset\}$ are dense. 
\end{itemize}
\end{lem}

In view of (1) and (2) in Lemma \ref{lem: facts}, (1) and (2) in Definition \ref{dfn: vis1} are equivalent as
\begin{itemize}
    \item[(1)] for any $\psi\in C^1$, if $x$ is a local maximum for $V-\psi$, then $F(x, \psi(x), d\psi(x))\leq 0$;  
    \item[(2)] for any $\psi\in C^1$, if $x$ is a local minimum for $V-\psi$, then $F(x, \psi(x), d\psi(x))\geq 0$. 
\end{itemize}

\begin{dfn}
Suppose $V$ is locally Lipschitz. Define the classical upper and lower Dini (directional) derivatives, respectively, as
\begin{subequations}
\begin{align}
& D^+V(x; p)=\limsup_{t\ra 0^+}\frac{V(x+tp)-V(x)}{t}\label{E: compare2a},\\
& D^-V(x; q)=\liminf_{t\ra 0^+}\frac{V(x+tq)-V(x)}{t}\label{E: compare2b}.
\end{align}
\end{subequations}
\end{dfn}

The following theorem \citep[Theorem 2.40, Chapter III, ][]{bardi1997optimal} states the equivalence of Dini solutions to viscosity solutions to GHJBs. We rephrase the theorem as follows. We omit the proof due to the similarity. 

\begin{thm}\label{thm: dini_directional}
Suppose $\Omega$ is a bounded open set, and $V\in C(\bar{\Omega})$. For GHJB $G(x, u, DV(x))=0$ with a fixed $u\in\uu$, the following statements are equivalent:
\begin{enumerate}
    \item[(1)] $-G(x, u, p)\leq 0$ for all $x\in\R^n$ and for all $p\in\partial^+V(x)$ (respectively, $\geq 0$);
    \item[(2)] $-L(x,u) - D^+u(x; f(x)+g(x)u)\leq 0$ for all $x\in\R^n$ (respectively, $\geq 0$). 
\end{enumerate}
\end{thm}

\section{Proofs in Section \ref{sec:prob}}\label{app: proof_sec_2}
\textbf{Proof of Proposition \ref{prop: DPP}:}
Note that for all $t>0$ and $u\in\uu$, we have
\begin{equation*}
    \begin{split}
        J(x,u) 
        =&\int_0^t\lk(\phi(s;x,u), u(s))ds+\int_t^\infty \lk(\phi(s;x,u), u(s)) ds\\
        =&\int_0^t\lk(\phi(s;x,u), u(s))ds+\int_0^\infty \lk(\phi(s+t;x,u), u(s+t)) ds\\
        =&\int_0^t\lk(\phi(s;x,u), u(s))ds+J(\phi(t;x,u), u')\\
         \geq &\inf_{u\in\uu}\left\{\int_0^t \lk(\phi(s;x,u))ds + V(\phi(t;x,u))\right\},
    \end{split}
\end{equation*}
where the controller $u'$ is defined as $u'(s):=u(s+t)$ for all $s>0$. 
Taking the infimum over $\uu$, we have 
$$V(x)\geq \inf_{u\in\uu}\left\{\int_0^t \lk(\phi(s;x,u), u(s))ds + V(\phi(t;x,u))\right\}.$$
Not we fix a $u\in\uu$, an $\eps>0$, and choose a $u'\in\uu$  such that
$$V(\phi(t;x,u))\geq J (\phi(t;x,u), u')-\eps.$$
Let the controller $u''$ be such that
$$u''(s)=\left\{\begin{array}{lr} 
u(s),\qquad s\leq t,\\
u'(s-t), \;s>t.
\end{array}\right.  $$
Then, 
\begin{equation*}
    \begin{split}
        V(x) \leq & J(x, u'')\\
         = &\int_0^t \lk (\phi(s;x,u), u(s))ds + \int_t^\infty \lk (\phi(s;x,u''), u''(s))ds\\
          = &\int_0^t \lk (\phi(s;x,u), u(s))ds + \int_0^\infty \lk (\phi(s;\phi(t;x,u),u'), u'(s))ds\\
         \leq &\int_0^t \lk (\phi(s;x,u), u(s))ds + V(\phi(t;x,u)) + \eps.
    \end{split}
\end{equation*}
Since $u$ and $\eps$ are given arbitrarily, by sending $\eps\ra 0$ and taking the infimum over $\uu$,
we have 
$$V(x)\leq  \inf_{u\in\uu}\left\{\int_0^t \lk(\phi(s;x,u), u(s))ds + V(\phi(t;x,u))\right\},$$
which completes the proof. \qed

\textbf{Proof of Proposition \ref{prop: uniqueness_H}}: 
We first show that $V$ is a viscosity solution using the equivalent conditions introduced in Appendix \ref{app: vis}. 
Let $\psi\in C^1$ and $x$ be a local maximum point of $V-\psi$. Then 
$$V(x)-V(z)\geq \psi(x)-\psi(z),\;\;\forall z\in\oball(x,r),$$
where $\oball(x,r)$ denotes the set $\{z\in\R^n: |z-x|<r\}$. 
Consider any constant control signal $u$. For $t$ sufficiently small, we have $\phi(t;x,u)\in\oball(x,r)$. Therefore,

\begin{equation}\label{E: small_t}
    \begin{split}
    &\psi(x)-\psi(\phi(t;x,u))\\
    \leq &V(x)- V(\phi(t;x,u))\\
        \leq &\int_0^t \lk(\phi(s;x,u), u(s))ds + V(\phi(t;x,u)) - V(\phi(t;x,u))
    \end{split}
\end{equation}
where the second line is in virtue of Proposition  \ref{prop: DPP}. Considering the infinitesimal behavior on both sides of \eqref{E: small_t}, we have
$$-D\psi(x)\cdot(f(x)+g(x)u)\leq \lk(x, u),$$
which implies that $H(x,D\psi)\leq  0$.

Now we verify the case when $x$ is a local minimum of $V-\psi$. For each $\eps>0$ and $t>0$, by the second part of Proposition \ref{prop: DPP}, 
there exists a $u''\in\uu$ such that 
\begin{equation*}
    \begin{split}
        V(x) & \geq \int_0^t\lk(\phi(s;x,u''), u''(s))ds + V(\phi(t;x,u))-t\eps\\
        & \geq \int_0^t\lk(x, u'')ds + V(\phi(t;x,u))-t\eps+ \mathcal{O}(t),\\
    \end{split}
\end{equation*}
where the second line is by the Lipschitz continuity of $\lk$ and $\phi$, and  $\mathcal{O}(t)/t\ra 0$ as $t\ra 0$. Therefore, by the local minimum property of $V-\psi$,
\begin{equation*}
\begin{split}
        \psi(x)-\psi(\phi(t;x,u''))&\geq V(x)-V(\phi(t;x,u''))\\
        & \geq \int_0^t\lk(x, u'')ds -t\eps+ \mathcal{O}(t). 
\end{split}
\end{equation*}
 Note that $u''$ is selected based on some arbitrary $\eps$ and $t$. Now one can use a similar infinitesimal argument as the first part and obtain
$H(x, D\psi(x))\geq 0.$

It can be easily verified that $|H(x,0)|\leq M$ for some constant $M$, and, for every $r$, there exists an $L_r$ such that $|H(x,p)-H(y,p)\leq L_r|x-y|(1+|p|)$ for $|x|, |y|\leq r$. 
The uniqueness argument follows \citep[Theorem 1, Section 10.2, ][]{evans2010partial}, \citep[Section VI.3, ][]{bardi1997optimal}, and \citep[Section 3]{camilli2001generalization}.  \qed

\textbf{Proof of Theorem \ref{thm: optimal_feedback}}:
  Let $p_i\in \partial^+V(x_i)$. Then, by Lemma \ref{lem: facts}, there exists
$\psi_i\in C^1$ such that $D\psi_i(x_i)=p_i$,  $V(x_i)=\psi_i(x_i)$, and $V\leq \psi$ in the neighborhood. Since $\{x\in\dd: \;\partial^+(x)\neq \emptyset\}$ is dense, setting $t_0=0$ and $t_k\ra\infty$ as $k\ra\infty$, we can patch up $J$ in the following sense given that $t_i-t_{i-1}>0$ is sufficient small for all $i\in\{1,2,\cdots\}$:
\begin{equation}
    \begin{split}
        J(x, \kappa(\phi(\cdot))) 
         = &\lim_{k\ra\infty}\left\{\sum_{i=0}^k \int_{t_i}^{t_k} \lk (\phi(s; \phi(t_i;x,\kappa), \kappa(\phi(s))ds 
         + J(\phi(t_k;x, \kappa),\kappa)\right\}\\
         \geq &\psi(x)-\lim_{k\ra\infty}\psi(\phi(t_k;x,\kappa)
         \geq  V(x). 
    \end{split}
\end{equation}
The opposite inequality can be shown in a similar manner. \qed

\section{Proofs in Section \ref{sec:convergence}}\label{app: proof_sec_4}
\subsection{Results in Section \ref{sec: conv_exact}}
\textbf{Proof of Proposition \ref{prop: uniqueness_GHJB}}: 
The existence and uniqueness of $\phi$ follows by the basic assumptions on \eqref{E:sys}. Now we define 
\begin{equation}
    V(x)=\int_0^\infty L(\phi(s; x, u), u(\phi(s; x, u))ds.
\end{equation}
Then $V$ is positive definite, and, for all $t>0$, we have 
\begin{equation*}
    \begin{split}
        V(x) = &  \int_0^t L(\phi(s; x, u), u(\phi(s; x, u))ds  + \int_t^\infty L(\phi(s; x, u), u(\phi(s; x, u))ds\\
         = & \int_0^t L(\phi(s; x, u), u(\phi(s; x, u))ds  + \int_0^\infty L(\phi(s; \phi(t; x, u), u), u(\phi(s; \phi(t; x, u), u))ds\\
         = & \int_0^t L(\phi(s; x, u), u(\phi(s))ds + V(\phi(t; x, u)).
    \end{split}
\end{equation*}
To verify $V$ is a viscosity solution, we use a similar method as in the proof of Proposition \ref{prop: uniqueness_H}. Let $\psi\in C^1$ and $x$ be a local maximum point of $V-\psi$. Then, 
for $t$ sufficiently small, we have $\phi(t;x,u)\in\oball(x,r)$. Therefore,

\begin{equation}\label{E: small_t2}
    \begin{split}
    &\psi(x)-\psi(\phi(t;x,u))\\
    \leq &V(x)- V(\phi(t;x,u))\\
        =& \int_0^t L(\phi(s; x, u), u(\phi(s))ds + V(\phi(t; x, u))- V(\phi(t;x,u)).
    \end{split}
\end{equation}
 Considering the infinitesimal behavior on both sides of \eqref{E: small_t2}, we have
$$-D\psi(x)\cdot(f(x)+g(x)u)\leq \lk(x, u),$$
which implies that $G(x, u, D\psi)\leq  0$. The other side of the comparison falls in the same procedure. 

To validate the uniqueness, we notice that for a stabilizing state feedback control $u=\kappa(x)$, the Lipschitz continuous function $f(x)+g(x)\kappa(x)$ can only have a zero at $\zero$. Given that $V(0)=0$, and suppose that $\Omega\subseteq\R$, the uniqueness is followed by \citep[Proposition 2]{liu2023physics} and \citep[Theorem 19]{meng2023learning}. In addition, we notice that the quantity $DV(x)=-L(x,\kappa(x))/(f(x)+g(x)\kappa(x))$ is differentiable continuous solution other than $0$. For higher dimensional case, we address the problem using the well-known method of characteristics. The solvability of $C^1$ solution depends on the non-singularity of the Jacobian matrix, 
which is only problematic at $\zero$. By the continuity of viscosity solution, $V$ is uniquely defined on $\Omega$. \qed

\begin{rem}

During the policy iteration, we cannot guarantee an everywhere $C^1$ property of the solutions to GHJBs. 
Revisiting the example in Section \ref{sec: exact_PI}, the corresponding GHJB is given by $x^2+ DV(x)\cdot (xu)=0$. One can simply initialize with an  admissible controller $u=\kappa_0(x)=-\frac{1}{2}|x|$. Then,  $V_0(x)=2|x|$ uniquely solves the GHJB only in the viscosity sense, which fails to be everywhere $C^1$ in the first iteration. 
In view of the last part of the above proof, the non-differentiable points are only decided by the zeros of $f(x)+g(x)\kappa(x)$. In addition, to resolve Remark \ref{rem: zeros}, one can simply follow the exact argument. 

\end{rem}

\textbf{Proof of Theorem \ref{thm: convergence}}: 
\begin{enumerate}
\item[(1)]By Proposition \ref{prop: uniqueness_GHJB}, the function $V\in C^1(\Omega\setminus\zero)\cap C(\Omega)$. Apart from $\zero$, the exact-PI provides an admissible controller for each iteration, which follows the exact procedure as in \citep[Lemma 5.2.4, ][]{beard1995improving}. At $\zero$, by exact-PI, the state feedback controller returns a $\zero$. In addition, the upper and lower Dini derivatives $D^+V(x;p)$ and $D^-V(x;p)$ for any $p$ exist and are bounded. The Lipschitz continuity of $u_i$ at $\zero$ also follows by the definition.
    \item[(2)] Note that, at differentiable points, the proof follows exactly as \citep[Theorem 3.1.4, ][]{jiang2017robust}. However, in line with the concept of viscosity solutions, we provide a general proof. To begin with,  along the trajectory subject to the controller $u_{i+1}$, we have, for each $x$, 
    \begin{equation}\label{E: dini_ineq}
    \begin{split}
               &V_{i+1}-V_i \\
               \leq &\int_0^\infty D^+(V_{i+1}-V_i)(\phi(s); f(\phi(s))+g(\phi(s))u_{i+1})ds\\
                \leq &\int_0^\infty (D^+V_{i+1}-D^-V_i)(\phi(s); f(\phi(s))ds+\int_0^\infty g(\phi(s))u_{i+1})ds,
    \end{split}
    \end{equation}
    where the existence of Dini derivatives are granted by the Lipschitz continuity of $V_i$ and $V_{i+1}$. 
    On the other hand,  $V_i$ and $V_{i+1}$ are, respectively, the unique viscosity solution to $G(x, u_i, DV_i)$ and $G(x, u_{i+1}, DV_{i+1})=0$.  By Theorem \ref{thm: dini_directional}, and plugging the directional Dini derivative $D^+V_{i+1}(\phi(s); f(\phi(s))+g(\phi(s))u_{i+1})$ and $D^-V_i(\phi(s); f(\phi(s))+g(\phi(s))u_{i}$ into \eqref{E: dini_ineq}, one can obtain that 
    \begin{equation}
        \begin{split}
            & V_{i+1}(x)-V_i(x)\\
            \leq &-\int_0^\infty \normr{u_i}^2+\normr{u_{i+1}}^2-2u_{i+1}^TRu_i ds\\
            \leq & -\int_0^\infty \normr{u_{i+1}-u_i}^2 ds\leq 0.
        \end{split}
    \end{equation}
    
    \item[(3)] It is a well-known result that a monotonic sequence of functions $\{V_i\}_{i\geq 0}$ that is bounded from below converges pointwise to a function $V_\infty$. In view of Dini's theorem (see also the proof of \citep[Theorem 5.3.1, ][]{beard1995improving}), the sequence also converges uniformly provided that $\Omega$ is compact. In addition, it can be verified that $\{V_i\}$ has a uniformly bounded Lipschitz constant on $\Omega$ and forms a compact subspace in $C(\Omega)$, which implies the Lipschitz continuity of $V_\infty$. The pointwise convergence of $\{u_i\}_{i\geq 0}$ also follows \citep[Theorem 3.1.4, ][]{jiang2017robust}. 
    
    It suffices to show that $V_\infty$ is a viscosity solution to \eqref{E: HJB_optimal}. On %
    $\Omega\setminus\{\zero\}$, $DV_\infty$ exists uniquely almost everywhere (a.e.) and is Lebesgue integrable. Based on the definition of $G_i$ as well as the pointwise convergence of $\{u_i\}$, we have that $DV_i$ is Lebesgue integrable for each $i$ and converges pointwise. Combining the fact that $V_i\ra V_\infty$ uniformly, we have  $\lim_{i\ra\infty}DV_i = DV_\infty$ a.e. by Radon-Nikodym theorem. However, it can be verified that $\lim_{i\ra\infty}DV_i$ solves \eqref{E: HJB_optimal} pointwise on $\Omega\setminus\{\zero\}$. It follows that $DV_\infty$ solves \eqref{E: HJB_optimal} a.e. on $\Omega\setminus\{\zero\}$. %
    At $\zero$,  we have that $V_\infty(\zero)=\lim_{i\ra\infty}V_i(\zero)=0$, which may not be differentiable. However, it is clear that $H(x,p)$ is convex in $p$ for each fixed $x$, and $H(x, DV_\infty(x))=0$ a.e. on $\Omega$.   By \citep[Proposition 5.2, Chapter II, ][]{bardi1997optimal}, $V_\infty$ is a viscosity solution on $\Omega$. In virtue of Proposition \ref{prop: uniqueness_H}, $V_\infty$ should also be the unique viscosity solution. Therefore, $V_\infty=V^*$. \qed
\end{enumerate}

\begin{rem}
It is worth noting that  (3) of  \citep[Theorem 3.1.4, ][]{jiang2017robust} is based on the  assumptions  $V^*, V_\infty\in C^1(\Omega)$. However, both assumptions are not necessarily guaranteed. As pointed out in Section \ref{sec: exact_PI}, the HJB 
$-x^2+\frac{1}{4}(DV(x))^2x^2=0$ has a unique viscosity solution $V^*(x)=2|x|$, which fails to be differentiable at $0$. As for $V_\infty$, it is the limit of $\{V_i\}$ only w.r.t. the uniform norm rather than the $C^1$-norm. The $C^1$ property of $V_\infty$ on $\Omega\setminus\{\zero\}$ is not even guaranteed.  

The convergence of $DV_i\ra DV_\infty$ (if exists) is also not clear in \citep[Theorem 3.1.4, ][]{jiang2017robust} and \citep[Theorem 5.3.1, ][]{beard1995improving}  relying only on the uniform convergence of $\{V_i\}$. To ensure that $V_\infty$ solves \eqref{E: HJB_optimal}, the work  \citep[Theorem 3.2, ][]{farsi2023reinforcement} made a strong assumption on the uniform convergence of $\{DV_i\}$, which cannot be guaranteed in practice.     In this view, it is necessary to consider the convergence and solutions in the viscosity sense for general cases. The above proof takes advantage of the convexity of $H(x, p)$ in $p$ such that only the property of $\lim_{i\ra\infty}DV_i=DV_\infty$ a.e. is needed. \qed

\end{rem}

\subsection{Results in Section \ref{sec: conv_PINN}}

In order to ensure the coherence of the proofs within this subsection, we recall the following notation. Given $\kappa_0$,  $\{V_i\}$ and $\{\kappa_{i+1}\}$ are updated by exact-PI. In other words, for each $i\geq 0$, $V_i$ is the unique viscosity solution to
\begin{equation}\label{E: exact_proof}
    \begin{split}
       &G_i(x, \kappa_i, DV_i(x)) \\
       =& Q(x)+\normr{\kappa_i(x)}^2+  DV_i(x)\cdot (f(x)+g(x)\kappa_i(x)) \\
        = & 0, 
    \end{split}
\end{equation}
and $\kappa_{i+1}$ is updated by \eqref{eq:pi_improve}. 
In addition,  $\{\hv_i\}$ and $\{\hk_{i+1}\}$ are updated by PINN-PI with $\hk_0=\kappa_0$. 
For each $i\geq 0$, we also denote  $\tV_i$ as the true viscosity solutions to
\begin{equation}\label{E: hat_tilde}
    \begin{split}
       & G_i(x, \hk_i, D\tV_i(x)) \\
        = &Q(x)+ \normr{\hk_i(x)}^2 +  D\tV_i(x)\cdot (f(x)+g(x)\hk_i(x)) \\
       = & 0. 
    \end{split}
\end{equation}
Accordingly, we also set $\tilde{\kappa}_{i+1}(x)=-\frac{1}{2}R^{-1}(x)g(x)D\tV_i(x)$ if $x\neq \zero$, and $\tilde{\kappa}_{i+1}(\zero)=\zero$.

Before proving  Theorem \ref{thm: conv}, we look at a lemma that entails the expected convergence result. 

\begin{lem}\label{lem: ideal}
Suppose for each $i\geq 0$, we have 
\begin{equation}\label{E: approx_PINN}
    \sup_{x\in \Omega } \norm{\hv_i(x)-\tV_i(x)} + \sup_{x\in \Omega }\norm{\hk_{i+1}(x)-\tilde{\kappa}_{i+1}(x)}\ra 0
\end{equation}
Then, PINN-PI can guarantee that
\begin{equation}\label{E: conv_PINN_ideal}
    \sup_{x\in \Omega } \norm{\hv_i(x)-V_i(x)} + \sup_{x\in \Omega }\norm{\hk_{i+1}(x)-\kappa_{i+1}(x)}\ra 0.
\end{equation}
\end{lem}

\begin{proof}
We prove the convergence by induction. For $i=0$, we have $\hk_0=\kappa_0$. Then $V_0=\tV_0$ by Corollary \ref{cor: iteration}. By the definition of $\kappa_1$ and $\tilde{\kappa}_1$, it follows that $\kappa_1=\tilde{\kappa}_1$.  We can train the neural network sufficiently well,  such that $\sup_{x\in \Omega } \norm{\hv_0(x)-\tV_0(x)} + \sup_{x\in \Omega }\norm{\hk_1(x)-\tilde{\kappa}_1(x)}\ra 0$, which immediately implies $$\sup_{x\in \Omega } \norm{\hv_0(x)-V_0(x)} + \sup_{x\in \Omega }\norm{\hk_{1}(x)-\kappa_{1}(x)}\ra 0.$$

For each $i\geq 1$, let $W_i=V_i-\Tilde{V}_i$. Then, $W_i\in C^1(\Omega\setminus\{\zero\}))$, $DW_i\in C(\Omega\setminus\{\zero\}))$, and $W_i(\zero)=0$. Since on $\Omega\setminus\{\zero\}$,   $V_i$ and $\tV_i$ solves \eqref{E: exact_proof} and \eqref{E: hat_tilde} in the conventional sense. A direct comparison of \eqref{E: exact_proof} and \eqref{E: hat_tilde} gives that 
\begin{equation*}
\begin{split}
        & DW_i(x)\cdot (f(x)+ g(x)\hk_i(x))  \\
        = & -g^T(x)DV_i(x) \cdot(\hk_i(x)-\kappa_i(x)))- \normr{\hk_i(x)}^2 +\normr{\kappa_i(x)}^2.\\
\end{split}
\end{equation*}
However, 
$g^TDV_i=-2R\kappa_i+2R(\kappa_i-\kappa_{i+1})$. Therefore, on $\Omega\setminus\{\zero\}$, 
\begin{equation*}
\begin{split}
        & DW_i(x)\cdot (f(x)+g(x)\hk_i(x)) \\ = & 2R(x)\kappa_i(\hk_i(x)-\kappa(x))-2R(x)(\kappa_i(x)-\kappa_{i+1}(x))(\hk_i(x)-\kappa_i(x)) - \hk_i^T(x) R(x) \hk_i(x) + \kappa_i^T(x) R(x) \kappa_i(x)\\
        = & -\normr{\hk_i(x)-\kappa_i(x)}^2  -2R(x)(\kappa_i(x)-\kappa_{i+1}(x))(\hk_i(x)-\kappa_i(x)),
\end{split}
\end{equation*}
and 
\begin{equation}
    \begin{split}
        & |DW_i(x)\cdot (f(x)+g(x)\hk_i(x)) | \\
        \leq & \normr{\hk_i(x)-\kappa_i(x)}^2 + 2\|R(x)\||\kappa_i(x)-\kappa_{i+1}(x) ||\hk_i(x)-\kappa_i(x) |.
    \end{split}
\end{equation}
For simplicity, we define an  intermediate Hamiltonian
\begin{equation}
    \begin{split}
        & F_i(x, DW_i(x)) \\
        := & |DW_i(x)\cdot (f(x)+g(x)\hk_i(x)) | -\|\hk_i(x)-\kappa_i(x)\|_R^2  - 2\|R\||\kappa_i(x)-\kappa_{i+1}(x)||\hk_i(x)-\kappa_i(x) |.
    \end{split}
\end{equation}
Note that for each $i\geq 1$, the mapping
$p \mapsto F_i(x, p)$
is convex for any fixed $x$. By \citep[Proposition 5.1, Chapter II, ][]{bardi1997optimal}, $W_i$ is the viscosity solution to $F_i(x, DW_i(x))=0$ on $\Omega$. Given that $\hk_i$ converges to $\kappa_i$ uniformly, applying \citep[Proposition 2.2, Chapter II, ][]{bardi1997optimal}, $W_i$ should uniformly converges (on $\Omega$) to the viscosity solution of 
$$|DW_i(x)|\cdot (f(x)+g(x)\hk_i(x)) = 0,$$ 
which is the constant 0. As a byproduct, due to the continuous differentiability on $\Omega\setminus\{\zero\}$, one can check that $D\tV_i \ra DV_i$ (or $|DW_i|\ra 0$) uniformly on $\Omega\setminus\{\zero\}$. However, by the definition of $\tilde{\kappa}_i$ and $\kappa_i$ again, we have  $\tilde{\kappa}_{i+1}\ra\kappa_{i+1}$ in the same sense on $\Omega$. By the hypothesis \eqref{E: approx_PINN}, and a triangle inequality argument, the uniform convergence in \eqref{E: conv_PINN_ideal} follows. 
\end{proof}

\begin{rem}\label{rem: conv_require}
In the proof, leveraging the convexity of $F_i(x, DW_i(x))$ in $DW_i(x)$, the result in \citep[Proposition 2.2, Chapter II, ][]{bardi1997optimal} ensures that the uniform convergence $|\tV_i-V_i|$ in $C^1$ norm on $\Omega\setminus\{\zero\}$ (almost everywhere) can lead to a weaker convergence on $\Omega$ (everywhere):
\begin{equation}\label{E: weaker}
    \sup_{x\in \Omega } \norm{\tV_i(x)-V_i(x)} + \sup_{x\in \Omega }\norm{\tilde{\kappa}_{i+1}(x)-\kappa_{i+1}(x)}\ra 0.
\end{equation}
In other words, we sacrifice the exact convergence of $|D\tV_i-DV_i|$ at $\zero$ (due to the potential lack of differentiability),  and use the convergence of $\sup_{x\in \Omega }\norm{\tilde{\kappa}_{i+1}(x)-\kappa_{i+1}(x)}$ instead. Note that, the mapping $g$ in the definition of $\tilde{\kappa}_{i+1}$ and $\kappa_i$ has a smoothing effect. And this is also the reason why we can achieve the desired convergence property on the entire $\Omega$. \qed
\end{rem}

\textbf{Proof of Proposition \ref{cor: pinn}}: 
Recall that this proposition considers the general case for any GHJB $G(x,\kappa(x), DV(x))=0$ with admissible $\kappa$ as in Proposition \ref{prop: uniqueness_GHJB}. Given that $f(x)+g(x)\kappa(x)$ is bounded from above and away from $\zero$ on $\Omega\setminus U_\eps$, it can be easily shown that $DV$ of the true solution $V$ is also Lipschitz continuous. 
We first introduce short hand notations $|\cdot|_\infty:=\sup_{x\in \Omega\setminus U_\eps }|\cdot|$ and $\|\cdot\|_2:=\int_{\Omega\setminus U_\eps} |\cdot|^2dx$ for functions. For any Lipschitz continuous function $h$ on $\Omega\setminus U_\eps$, we define the Lipschitz constant as
$$\Lip(h):=\sup_{x\neq y}\frac{|h(x)-h(y)|}{|x-y|}. $$ For any $\hv_N\in\ff$, let $\rr_N(x):= G(x, \kappa(x), D\hv_N(x))$. It is clear that $\rr_N$ is Lipschitz continuous by the definition of $G$. \\
\noindent\textbf{Step 1:} We first show some useful bounds. 
Note that, given the compactness of $\Omega\setminus U_\eps$ and the continuous differentiability of $\hv_N$, we have
\begin{equation}\label{E: b1}
\begin{split}
   |\hv_N|_{C^1} &=
     |\hv_N|_\infty+|D\hv_N|_\infty\\
     &\leq C_1\cdot |D\hv_N|_\infty + |\hv_N(\zero)|, 
\end{split}
\end{equation}
where $C_1=\sup_{x\in \Omega\setminus U_\eps } \norm{x} + 1$. In addition, since $f(x)+g(x)\kappa(x)$ is bounded from above and away from $\zero$ on $\Omega\setminus U_\eps$, it can be verified that 
\begin{equation}\label{E: b2}
    \begin{split}
        C_2\cdot \norms{D\hv_N}&\leq \sup_{x\in \Omega\setminus U_\eps }\norm{L(x,\kappa(x))+\rr_N(x)}\\
        &\leq C_3\cdot  \norms{DV_N}, 
    \end{split}
\end{equation}
where $C_2=\inf_{x\in\Omega\setminus U_\eps}|\min\{f(x)+g(x)\kappa(x)\}|$ and $C_3=\norms{f(x)+g(x)\kappa(x)}$. 

Now we show the bound for $\norms{\rr_N(x)}$. Let $x^*, x_*$ be the maximizer and minimizer for $\norm{\rr_N(x)}$, respectively. Then
\begin{equation*}
    \begin{split}
        &\int_{\Omega\setminus U_\eps} |\rr_N(x)|^2 dx\\
        \geq & \mu(\Omega\setminus U_\eps)\cdot |\rr_N(x_*)|^2\\
        = & \mu(\Omega\setminus U_\eps)\cdot |\rr_N(x_*)-\rr_N(x^*)+\rr_N(x^*)|^2\\
        \geq & \mu(\Omega\setminus U_\eps)\cdot(-2|\rr_N(x_*)-\rr_N(x^*)||\rr_N(x^*)|+\rr_N^2(x^*)),  
    \end{split}
\end{equation*}
and consequently, 
\begin{equation*}
    \begin{split}
       & \norms{\rr_N(x)}^2 \\
      \leq  &  \left(\frac{1}{\mu(\Omega\setminus U_\eps)}\|\rr_N\|_2^2\right)+2|\rr_N(x_*)-\rr_N(x^*)||\rr_N(x^*)|\\
       \leq  & \left(\frac{1}{\mu(\Omega\setminus U_\eps)}\|\rr_N\|_2^2\right)+4\Lip^2(\rr_N)\cdot\sup_{x\in \Omega\setminus U_\eps } |x|.
    \end{split}
\end{equation*}
This implies that there exists a $C_4>0$ such that 
\begin{equation}\label{E: b3}
     |\rr_N|_\infty\leq C_4(\|\rr_N\|_2+ \Lip(\rr_N)). 
\end{equation}
\noindent\textbf{Step 2:} We show the continuous dependence of $\|\rr_N\|_2$ on $E_{T,N}(\hv_N)$. Define
		\[
		\set{\hv_N\in \mathcal F: \Lip(\rr_N)+\norm{\hv_N(\zero)} < r} =:\mathcal F_{r},
		\]
which is uniformly equicontinuous and hence compact. 
Pick $\tta >0$. %
		By the uniform continuity of $G$,  there is a $\delta >0$ such that for all $\hv_N \in \mathcal{F}_r$ and every $x,z\in X$ with $\norm{x-z} <\delta$,  we have 
		\begin{equation*}
		    \begin{split}
		         |\rr_N(x)-\rr_N(y)|  <\min\set{\sqrt{\frac{\tta}3}, \frac \tta {6M_{\mathcal F_r}}}
		    \end{split}
		\end{equation*}
		where $M_{\mathcal F_r}$ is a uniform upper bound on $|\rr_N(x)|$ for $\hv_N\in \bar{\mathcal F_r}$ and $x\in \Omega\setminus U_\eps$ (this bound exists by compactness of  $\bar{\mathcal F_r}$ and continuity of $G$).
			
		For this $\delta$, we can pick an $N_1\in \N$ so that $\delta_{N_1}<\delta$. It follows that for all $N\geq N_1$, we have $\Omega\setminus U_\eps\subseteq \bigcup_{k=1}^N \overline{\B[X]{\delta_N}{x_k}}$ and $\delta_N <\delta$.  By Hypothesis \ref{hyp: lip} and the assumption that $E_{T,N}(\hv_N) \ra 0$,  there is a $N_2 \in \N$ with for all $n\geq N_2$ we have $E_{T,N}^{\text{mod}}(\hv_N) < \frac \tta {3}$, where
  \[
		E_{T,N}^{\text{mod}}(\hv_N)=\frac{1}{N} \sum_{k=1}^N{|\rr_N(x_k)|^2}
		+|\Lip(\rr_N)|+\norm{\hv(\zero)}.
		\]
Then, for all $N\geq \max\set{N_1,N_2}$
\begin{equation*}
\begin{split}
        \|\rr_N\|_2 & \leq \int_{x\in \bigcup_{k=1}^N\mathcal{B}_{\delta_N}^{\Omega\setminus U_\eps}(x_k)}\norm{\rr_N(x)}^2dx \\
        & \leq \sum_{k=1}^N\int_{x\in \mathcal{B}_{\delta_N}^{\Omega\setminus U_\eps}(x_k)}\norm{\rr_N(x)}^2dx\\
        & \leq \sum_{k=1}^N \mu\left(\mathcal{B}_{\delta_N}^{\Omega\setminus U_\eps}(x_k)\right)|\rr_N(x_k^*)|^2
\end{split}
\end{equation*}
		where $x_k^* \in \B[\Omega\setminus U_\eps]{\delta_N}{x_k}$ is the maximizer of  $\norm{\rr_N(x)}$. Let $\mu_N:=\mu\round{\B[\Omega\setminus U_\eps]{\delta_N}{x_k}}$.  
Then, by uniform continuity of $\rr_N$, we see
		\begin{align*}
			\sum_{k=1}^N \mu_N|\rr_N(x_k^*)|^2 
   \leq &\mu_N \sum_{k=1}^N \round{\norm{\rr_N(x_k^*)-\rr_N(x_k)} +\norm{\rr_N(x_k)}}^2 \\
			< &\mu_N \sum_{k=1}^N\round{ \min\set{\sqrt{\frac{\tta}3}, \frac \tta {6M_{\mathcal F_r}}} + \norm{\rr_N(x_k)}}^2 \\
			\leq &\mu_N \sum_{k=1}^N \round{\frac \tta 3 + \frac \tta 3 \frac {\norm{\rr_N(x_k)}} {M_{\mathcal F_r}} +\norm{\rr_N(x_k)}_2^2} \\
			\leq &\mu_N \sum_{k=1}^N \round{\frac \tta {3}+\frac \tta 3} +\mu_N \sum_{k=1}^N  \norm{\rr_N(x_k)}_2^2 \\
			\leq & \tta C_5.
		\end{align*}
		where $C_5:=N\mu_N$.  This completes the proof of Step 2.\\
  \noindent\textbf{Step 3:} Now we are ready to prove the statement in this proposition. For simplicity, 
  we write $a\lsim b$ if there exists a constant
$C>0$, independent of $a$ and $b$, such that $a\leq C b$. 
  
  We pick sufficiently large $N,M$, then, by Eq. \eqref{E: b1}, \eqref{E: b2}, and \eqref{E: b3}, 
  \begin{equation*}
      \begin{split}
          |\hv_N-\hv_M|_{C^1}\lsim& |\rr_N-\rr_M|_\infty+|\hv_N(\zero)-\hv_M(\zero)|\\
           \lsim &\|\rr_N-\rr_M\|_2+\Lip(\rr_N-\rr_M) +|\hv_N(\zero)-\hv_M(\zero)|.
      \end{split}
  \end{equation*}
In addition, given that $E_{T,N}(\hv_N)$ can be arbitrarily small, by Step 2 and Hypothesis \ref{hyp: lip}, it is  clear that $\{\hv_N\}_N$ forms a Cauchy sequence. We denote the limit as $\hv^*\in\ff$, which solves the equation $G(x, \kappa(x), D\hv^*(x))=0$. By the uniqueness of the solution of $G$ as in Proposition \ref{prop: uniqueness_GHJB}, one can conclude that $\hv^*(x)=V(x)$.

\textbf{Proof of Theorem \ref{thm: conv}}: 
For any $\eps>0$ and for each $i$, by Proposition \ref{cor: pinn}, one can 
guarantee a $C^1$-uniform convergence of $\hv_i$ to $\tV_i$ on the compact set $\Omega\setminus U_\eps$, where $\hv_i$ plays the role of the representable functions in Proposition \ref{cor: pinn}. In view of Remark \ref{rem: conv_require},  one can obtain the convergence  
\begin{equation}\label{E: conv_PINN_real_2}
    \sup_{x\in \Omega\setminus U_\eps } \norm{\hv_i(x)-\tV_i(x)} + \sup_{x\in \Omega\setminus U_\eps }\norm{\hk_{i+1}(x)-\tk_{i+1}(x)}\ra 0,
\end{equation}
with a modified convergence region $\Omega\setminus U_\eps$ instead of $\Omega$. 

On $U_\eps\setminus\{\zero\}$, by the continuous differentiability of $\hv_i$, $\tV_i$, and $g$, one can immediately verify that 
\begin{equation}\label{E: V_eps}
    |\hv_i(x)-\tV_i(x)|\leq \sup_{x\in U_\eps} (|D\hv_i(x)|+|D\tV_i(x)|)\cdot|x|
\end{equation}
and 
\begin{equation}\label{E: kappa_eps}
    \begin{split}
        |\hk_{i+1}(x)-\tk_{i+1}(x)|\leq  C|x|, 
    \end{split}
\end{equation}
where $$C=\sup_{x\in U_\eps\setminus\{\zero\}}|Dg^T(x)|\cdot \sup_{x\in U_\eps\setminus\{\zero\}} (|D\hv_i(x)|+|D\tV_i(x)|)$$ and $Dg^T(x)$ is the Fr\'{e}chet derivative at $x$. Note that the quantities $\hv_i$, $\tV_i$, $\hk_{i+1}$, and $\tk_{i+1}$ are all null values at $\zero$. In addition, since $\eps>0$ is arbitrarily small, combining Eq. \eqref{E: conv_PINN_real_2}, \eqref{E: V_eps}, \eqref{E: kappa_eps}, one can have the arbitrarily smallness of 
$\sup_{x\in \Omega} \norm{\hv_i(x)-\tV_i(x)} + \sup_{x\in \Omega}\norm{\hk_{i+1}(x)-\tk_{i+1}(x)}$. By
Lemma \ref{lem: ideal}, the statement in Theorem \ref{thm: conv} follows immediately. \qed

\section{Further details on verification of Lyapunov conditions}\label{sec:lyap}

We aim to verify that a value function returned by Algorithms \ref{alg:elm-pi} or \ref{alg:pinn-pi} satisfies the Lyapunov condition (\ref{eq:dVeps}), recalled here as
\begin{equation}
    \label{eq:dVeps2}
    DV(x)(f+g(x)\kappa(x)) \le - \mu,\quad x\in \Omega \setminus U_\eps. 
\end{equation}
To obtain more information, we can verify a slightly stronger set of conditions defined below to facilitate the claim for asymptotic attraction and region of attraction:
\begin{enumerate}
    \item $DV(x)(f+g(x)\kappa(x)) \le - \mu,\quad\forall x\in \set{x\in \Omega:\,c_1\le V(x)\le c_2}$;
    
    \item $\set{x\in \Omega:\,V(x)\le c_2} \cap \partial \Omega = \emptyset$, where $\partial \Omega$ is the boundary of $\Omega$;
    
    \item $\set{x\in \Omega:\,V(x)\le c_1} \subset U_\eps$.
\end{enumerate}
By these conditions, $V$ will decrease along solutions of the closed-loop system under control $u=\kappa(x)$ and cannot escape $\Omega$ if solutions start in $\Omega_{c_2}:=\set{x\in \Omega:\,V(x)\le c_2}$. Furthermore, these solutions eventually reach the set $\Omega_{c_1}=\set{x\in \Omega:\,V(x)\le c_1}$, which is contained in $B_\eps$. While (\ref{eq:dVbasis}) appears to be a weaker condition than the above set of three conditions, if $V$ is positive definite on $\Omega$ with respect to the origin, then, for any $\eps>0$, we can always choose $c_1$ and $c_2$ such that the above conditions hold. Verifying these conditions readily gives a region of attraction $\Omega_{c_2}$ and an attractive set $\Omega_{c_1}$. 

\section{Numerical experiments}\label{sec:appendix:experiments}

\subsection{Implementation details}

All instances of ELM-PI and PINN-PI are run with the tanh activation function, unless otherwise noted. The tanh activation function is effective in approximating smooth functions. If non-smooth functions are involved, ReLU activation might be preferred (see Section \ref{sec:bilinear}). 

The same as in Section \ref{sec:examples} and Table \ref{tab:synthetic}, we run all examples with $N= m*d$, where $m$ is the size of the network, and $d$ is the dimension of the problem. The iteration number of PI is set to be 10. We only implemented a one-layer network for PINN-PI to draw fair comparisons with ELM-PI and basis function approaches such as successive Galerkin approximations (see Section \ref{sec:sga}).

For both ELM-PI and PINN-PI, the number of iterations in PI is set to be 10. For PINN-PI, we train the network for 10,000 steps in each iteration with Adam. For $m\le 1,600$, we conduct three separate runs of the experiments to report the average maximum testing errors and runtimes. For $m\ge 3,200$, the computational time is significantly larger and reported for a single run. ELM-PI experiments were run with an Intel Gold 6148 Skylake @ 2.4 GHz, and PINN-PI experiments were run with an NVidia V100SXM2 (16G memory). The errors reported are testing errors on $2*N$ points, where $N$ is the number of collocation points. 

\subsection{Verification of synthetic $n$-dimensional nonlinear control}

The domain on which we solve the problem is set to be $\Omega = [-1,1]^n$. Recall that the optimal value function is given by
$$
V^*(x)= \sum_{i=1}^n \left(\frac12 x_i^4 + \frac12  (x_i^2 + 1)^2 - \frac12\right).
$$ 

The detailed comparisons of training ELM-PI and PINN-PI on this synthetic example are shown in Table \ref{tab:synthetic}.

We also conducted formal verification experiments on this example. It is easy to verify that the largest level set of $V^*$ contained in $\Omega$ is $\set{V^*\le 2}$. In other words, we should be able to verify conditions (1)--(3) in Section \ref{sec:lyap} with any $0<c_1<c_2<2$. We verified value functions returned by ELM-PI and PINN-PI against conditions (1)--(3). The results are summarized in Table \ref{tab:verify_synthetic}. The parameters used for verification are $\mu=1e-4$, $c_1=0.01$, $c_2=1.99$, and $\eps = 0.1$. We employed dReal \cite{gao2013dreal} for verification. It is evident from the results that dReal, which utilizes interval analysis, falls prey to the curse of dimensionality. We also noted an intriguing observation that the value functions returned by ELM-PI, despite possessing the same form and number of neurons as those from PINN-PI, appear to pose a greater challenge for verification by dReal. This peculiar observation currently lacks a robust explanation, and it may pertain to the implementation of the dReal tool \cite{gao2013dreal}. On potential explanation is that the value function trained with ELM tends to have large coefficients, which may pose a challenge for verification with dReal. One may mitigate the issue by adding a L2 regularization term for the coefficients, at the expense of accuracy. We plan to investigate this issue further as well as alternative tools for verification in our future work.

\begin{table}
  \caption{Performance of ELM-PI and PINN-PI on a synthetic $n$-dimensional nonlinear problem: Here, $n$ represents the dimension of the problem, $m$ denotes the number of hidden units used for approximation, and $N$ indicates the number of collocation points.\\ %
  }
  \label{tab:synthetic}
  \centering
  \begin{tabular}{ccccccc}
    \toprule
    \multicolumn{3}{c}{Problem \& model size} & \multicolumn{2}{c}{ELM-PI} & \multicolumn{2}{c}{PINN-PI} \\
    \cmidrule(r){1-3} \cmidrule(lr){4-5} \cmidrule(l){6-7}
    $n$ & $m$ & $N$ & Error & Time (s)  & Error & Time (s) \\
    \midrule
    1 & 50 & 50 & 1.54E-10 & 0.03 & 2.45E-03 & 236.69 \\
    1 & 100 & 100 & 1.14E-10 & 0.08 & 2.86E-03 & 240.79 \\
    1 & 200 & 200 & 1.25E-11 & 0.21 & 2.17E-03 & 284.02 \\
    \midrule
    2 & 50 & 100 & 1.26E-01 & 0.07 & 1.03E-02 & 237.31 \\
    2 & 100 & 200 & 8.75E-04 & 0.18 & 4.97E-03 & 273.05 \\
    2 & 200 & 400 & 8.37E-07 & 0.52 & 9.87E-03 & 336.04 \\
    2 & 400 & 800 & 1.79E-08 & 1.81 & 1.69E-02 & 560.14\\
    2 & 800 & 1600 & 2.47E-09 & 25.94 & 3.10E-02 & 843.12  \\
    2 & 1600 & 3200 & 6.10E-10 & 184.98 & 1.53E-02 & 953.37 \\
    \midrule
    3 & 200 & 600 & 7.65E-02 & 0.82 & 1.95E-02 & 371.18 \\
    3 & 400 & 1200 & 7.66E-03 & 2.62 & 1.21E-02 & 762.42\\
    3 & 800 & 1600 & 1.34E-04 & 33.77 & 1.48E-02 & 866.44 \\
    3 & 1600 & 4800 & 2.18E-06 & 284.53 & 2.55E-02 & 1034.34 \\
    3 & 3200 & 9600 & 2.58E-07 & 3593.61 & 2.81E-02 & 768.84 \\
    \midrule
    4 & 800 & 3200 & 2.02E-01 & 24.59 & 2.85E-02 & 903.55\\
    4 & 1600 & 6400 & 3.29E-02 & 346.79 & 3.51E-02 & 1121.94 \\
    4 & 3200 & 12800 & 2.92E-03 & 4771.60 & 3.52E-02 & 984.04\\
    4 & 6400 & 25600 & 4.35E-05 & 43052.25 & 4.39E-02 & 3304.76\\
    \midrule
    5 & 800 & 4800 & 5.36E00 & 32.22 & 3.63E-02 & 770.54 \\
    5 & 3200 & 16000 & 3.08E-01 & 6303.20 & 5.20E-02 & 1204.22\\ 
    5 & 6400 & 32000 & 6.37E-02 & 53479.10 & 9.10E-02 & 5906.69 \\
    \midrule
    6 & 800 & 4800 & 8.76E00 & 39.03 & 4.31E-02 & 800.76 \\
    6 & 6400 & 38400 & 2.33E00 & 63403.53 & 1.38E-01 & 6995.10\\
    \midrule
    7 & 800 & 100000 & -- & -- & 3.74E-02 & 4414.28\\
    8 & 800 & 100000 & -- & -- & 5.28E-02 & 4415.12\\
    9 & 800 & 100000 & -- & -- & 4.88E-02 & 4422.54\\
    10 & 800 & 100000 & -- & -- &3.66E-02 & 4424.33\\
    11 & 800 & 100000 & -- & -- & 8.29E-02 & 4424.00\\
    12 & 800 & 100000 & -- & -- & 6.11E-02 & 4426.54 \\
    \bottomrule
  \end{tabular}
\end{table}

\begin{table}
  \caption{Training and verification of ELM-PI and PINN-PI on an $n$-dimensional nonlinear control example: value functions returned by ELM-PI and PINN-PI are verified against the Lyapunov conditions in Section \ref{sec:lyap}. The experimental setup is the same as the results in Table \ref{tab:synthetic}. The experiments in this table were run on a MacBook Pro with a 2 GHz Quad-Core Intel Core i5. Note that the results may slightly vary due to the selection of a random seed. The symbol $\times$ indicates that verification by dReal \cite{gao2013dreal} returned a counterexample for $c_1=0.01$ and $c_2=1.99$ and \text{t.o.} indicates verification was not conclusive within 1,800 (s). \\}
  \label{tab:verify_synthetic}
  \centering
  \begin{tabular}{cccc}
    \toprule
    \multicolumn{2}{c}{Problem \& model size} & 
    \multicolumn{1}{c}{ELM-PI}   & 
    \multicolumn{1}{c}{PINN-PI}    \\
    \cmidrule(r){1-2} \cmidrule(r){3-3} \cmidrule(r){4-4}
    $n$     & $m$  &  Verify (s) & Verify (s) \\
    \midrule
    1 & 50  & 6.99 & 0.02 \\
    1 & 100 & 2.20 & 0.05 \\
    1 & 200 & 4.13 & 0.14 \\
        \midrule
    2 & 50 & \text{t.o.} &  $\times $\\
    2 & 100 & \text{t.o.}&  0.68\\
    2 & 200 & \text{t.o.} &  1.88 \\
    2 & 400 & \text{t.o.} & 6.10\\
    \bottomrule
  \end{tabular}
\end{table}

\begin{figure}[ht]
    \centering
        \centering
         \subfigure[$m=100,\;N=200$]{\includegraphics[width=0.47\linewidth]{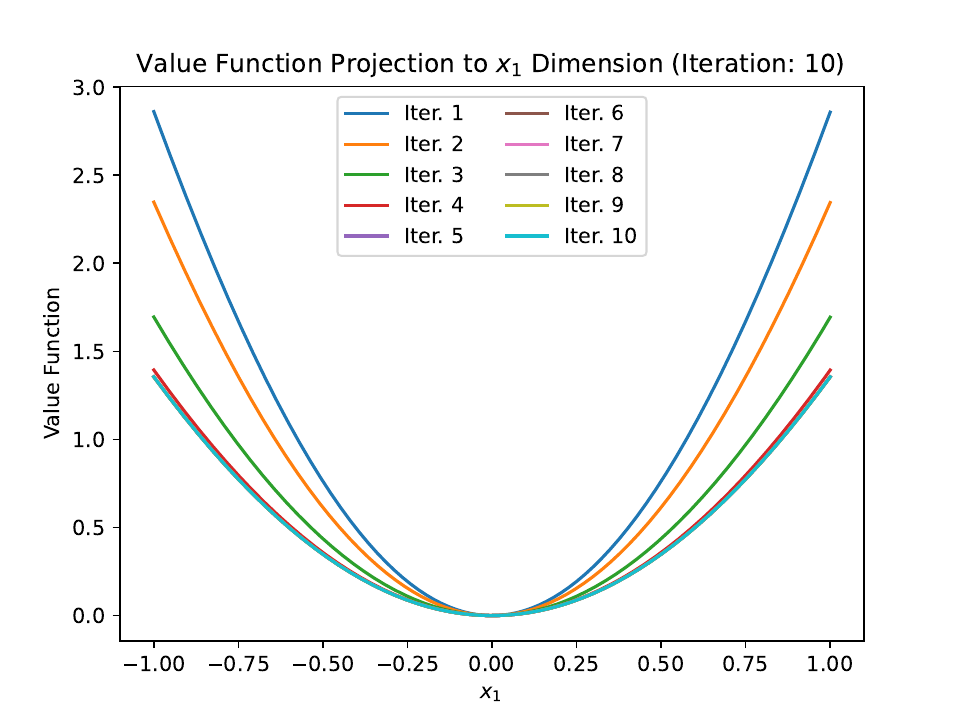}}
         \subfigure[$m=50,\;N=200$]{\includegraphics[width=0.47\linewidth]{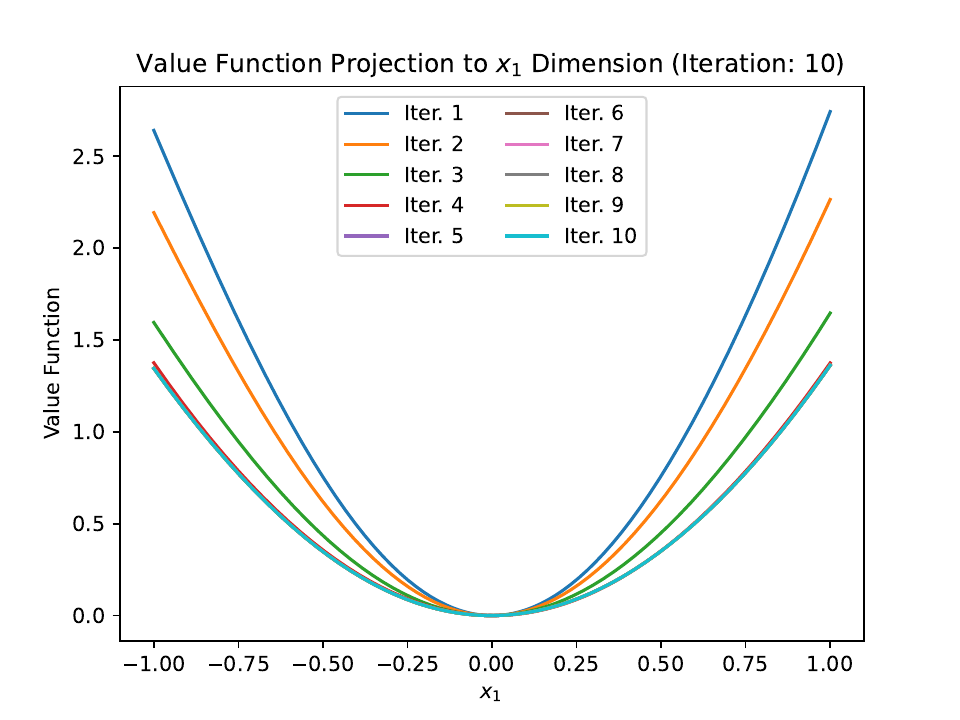}}
    \caption{ELM-PI on inverted pendulum: despite visual similarity and apparent convergence, the controller obtained from $m=50$ fails to stabilize the system, while the one from $m=100$ can be verified to be stabilizing using an SMT solver.}
    \label{fig:pendulum}
\end{figure}

\subsection{Inverted pendulum and comparison with successive Gakerlin approximations}\label{sec:sga}

The dynamics of the inverted pendulum are described by 
$
\ddot{\theta} = \frac{mg\ell \sin\theta - \mu\dot\theta + u}{m\ell^2}
$, where $u$ is the control input. We consider $\ell=0.5$, $m=0.1$, $g=9.8$, $\mu=0.1$. The cost function is defined by $Q=I_2$ and $R=2$. We run both ELM-PI and PINN-PI to compute the optimal control and policy. In this case, we do not have the analytical expression of the optimal value function as the ground truth. We extract the resulting optimal controllers and plot trajectories from random initial conditions to show the stability and performance of the controllers. 

We use the inverted pendulum example to compare successive Gakerlin approximations (SGA) \cite{beard1997galerkin} and the proposed neural policy iteration algorithms. We run ELM-PI and PINN-PI with $m=50$, $m=100$, $m=200$, and $m=400$. We also run SGA with polynomial bases 
of order 2, 4, 6, 8. All algorithms are run with 10 iterations. The results are summarized in Table \ref{tab:sga_compare}. From these results, it is clear that ELM-PI requires significantly less computational time on this low-dimensional example, as shown already in Section \ref{sec:examples} and Table \ref{tab:synthetic}. To demonstrate the performance of the obtained controllers, we simulate trajectories resulting from different controllers with random initial conditions with the same random seed for different controllers). Figure \ref{fig:sga_compare_cost} depicts the simulated costs averaged over 50 trajectories. It can be seen that high-order SGA achieves the same cost as ELM-PI with different $m$ values. PINN-PI achieves similar costs but requires longer training time as expected. 

We also verified the stability of the resulting controllers. In all the cases, we are able to verify the Lyapunov stability conditions outlined in Section \ref{sec:lyap} with $c_1=0.01$ and $c_2=0.029$. While $c_2=0.029$ appears to be small, it indeed gives the largest level set of the optimal value function contained in the region of interest, as shown in Figure \ref{fig:ROAs}. 

Furthermore, Figure \ref{fig:pendulum} depicts the training results for ELM-PI with $m=50$ and $m=100$. While the plots appear similar, the controller obtained from $m=50$ fails to stabilize the system,  while the one from $m=100$ can be verified to be stabilizing using an SMT solver. This highlights the importance of formal verification to ensure stability guarantees. 

\begin{figure}[h]
    \centering
        \includegraphics[width=0.47\textwidth]{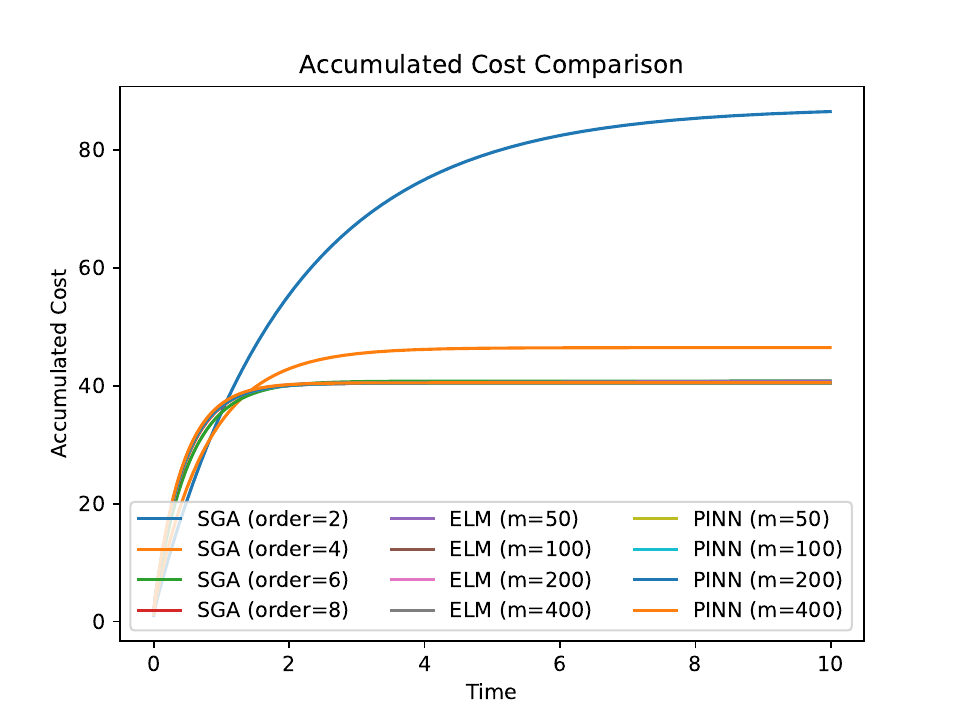}
    \caption{ELM-PI, PINN-PI, and SGA on the inverted pendulum example: it can be seen that the value returned by a high-order SGA achieves the same cost as ELM-PI with a different number of neurons, while the computational time required by ELM-PI is significantly less. In all the cases, we are able to verify the Lyapunov stability conditions outlined in Section \ref{sec:lyap} are met. }
    \label{fig:sga_compare_cost}
\end{figure}

\begin{figure}[h]
    \centering
        \includegraphics[width=0.47\textwidth]{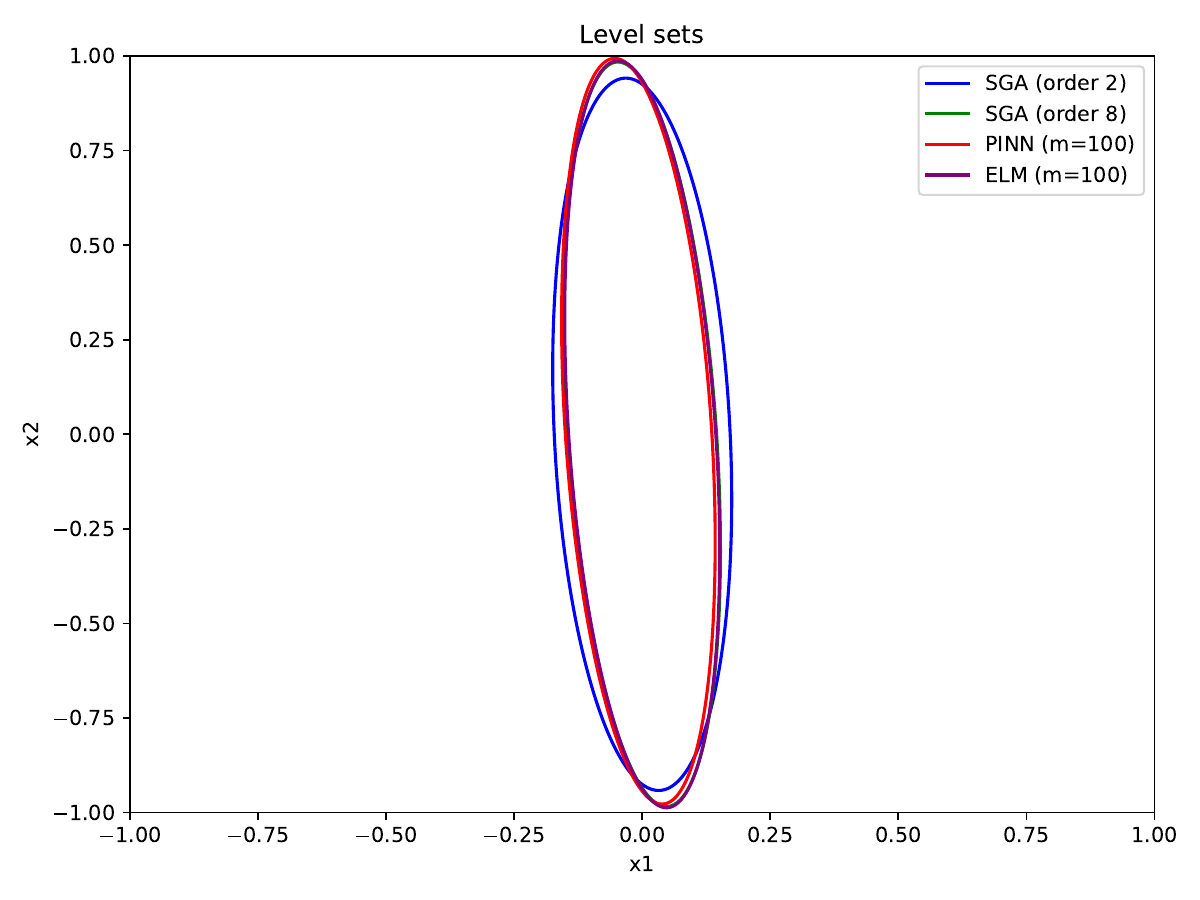}
    \caption{Certified regions of attraction by ELM-PI, PINN-PI, and SGA on the inverted pendulum example: it can be seen that for high-order SGA, PINN-PI, and ELM-PI, a region of attraction close to the boundary of the region of interest $\Omega$ can be verified using SMT solvers. }
    \label{fig:ROAs}
\end{figure}

\begin{table}
  \caption{Comparison of ELM-PI, PINN-PI, and SGA on the inverted pendulum example: we run ELM-PI and PINN-PI with $m=50$, $m=100$, $m=200$, and $m=400$
and SGA \cite{beard1997galerkin} with polynomial bases of order 2, 4, 6, 8. We record the training/computational time and whether the resulting controller is verifiably stabilizing.  \\}
  \label{tab:sga_compare}
  \centering
  \begin{tabular}{ccccccccc}
    \toprule
    \multicolumn{3}{c}{SGA} & 
    \multicolumn{3}{c}{ELM-PI}   & 
    \multicolumn{3}{c}{PINN-PI}    \\
    \cmidrule(r){1-3} \cmidrule(r){4-6} \cmidrule(r){7-9}
    Order & Time (s)   & Verified?   &  $m$   & Time (s) & Verified?  &  $m$   & Time (s) & Verified? \\
    \midrule
    2 & 4.80 & Yes & 50 & \textbf{0.11} & Yes & 50 & 255.15 & Yes \\
    4 & 19.37 & Yes & 100 & \textbf{0.24} & Yes & 100 & 256.53 & Yes \\
    6 & 66.52 & Yes & 200 & \textbf{0.71} & Yes  & 200 & 258.89 & Yes  \\
    8 & 212.42 & Yes & 400 & \textbf{2.92} & Yes & 400& 256.52 & Yes \\
    \bottomrule
  \end{tabular}
\end{table}

\section{Comparison with reinforcement learning algorithms on benchmark nonlinear control problems}\label{app:comparison}

\subsection{Rationale for algorithms selection in comparison}\label{sec:rationale}

We chose three recent RL algorithms as benchmarks: two model-free ones, PPO and HJBPPO, and one model-based, CT-MBRL. We used the implementation of PPO from the stable-baselines3 library \cite{stable-baselines3}, and we implemented HJBPPO by modifying this library. Since HJBPPO addresses the same problems as ours and incorporates the HJB equation to derive better policies than PPO, we compare our algorithm with both as model-free RL benchmarks. On the other hand, CT-MBRL is a state-of-the-art model-based RL algorithm that concerns similar environments to those in our numerical experiments.

It is worth mentioning that the ultimate goal of our algorithm is to devise optimal and stabilizing controllers for nonlinear systems after a few policy iterations, which can provide asymptotic stability for an infinite time horizon. This differs from the typical performance comparisons for RL algorithms, where success rate or normalized/averaged scores obtained using episodic training are used. For instance, for Cartpole, our PI-generated policy can maintain the rod in the upright position after a few seconds, while policies generated by most RL algorithms oscillate around the upright position. As long as it does not fall out of a given interval, it is regarded as a success with a reward, which means the trajectories do not asymptotically converge to the equilibrium points. In contrast, our algorithms achieve \textit{asymptotic} stability, meaning the controller can be deployed for any duration of time with convergence guarantees.

\subsection{Comparison results}
We compare PINN-PI with three reinforcement learning (RL) algorithms in Figure \ref{fig:acc_cost} by comparing their accumulated control costs over time. As shown in the plots, PINN-PI significantly outperforms the rest of the algorithms in the inverted pendulum environment and the remaining higher-dimensional environments. Code for these comparisons, along with other examples, is provided in the supplementary material.

\begin{figure}[H]
    \centering
     \subfigure[Inverted Pendulum]{\includegraphics[width=0.4\linewidth]{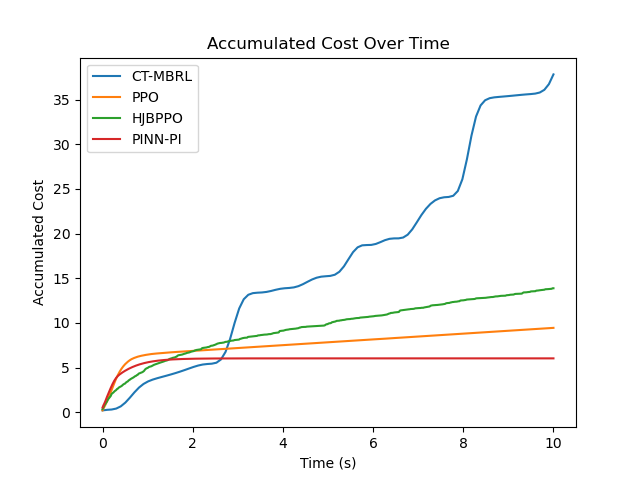}}
     \subfigure[Cartpole]{\includegraphics[width=0.4\linewidth]{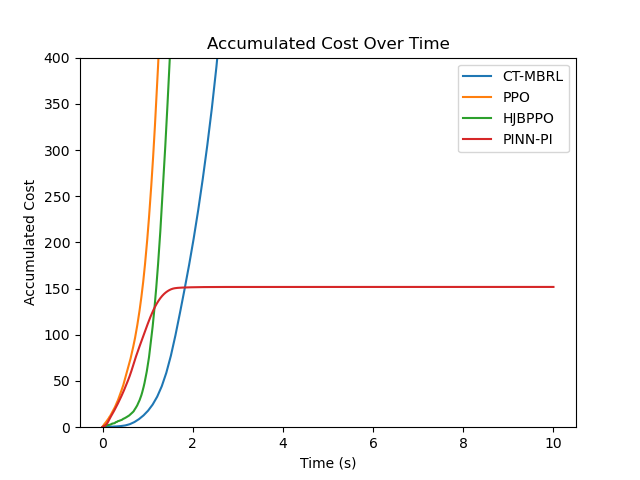}}
     \subfigure[2D Quadrotor]{\includegraphics[width=0.4\linewidth]{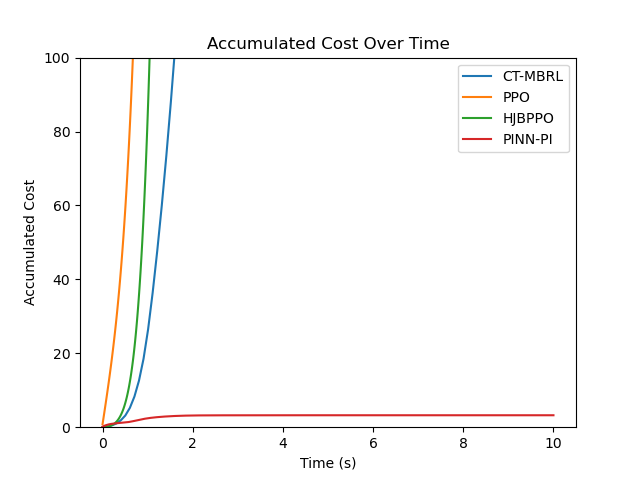}}
     \subfigure[3D Quadrotor]{\includegraphics[width=0.4\linewidth]{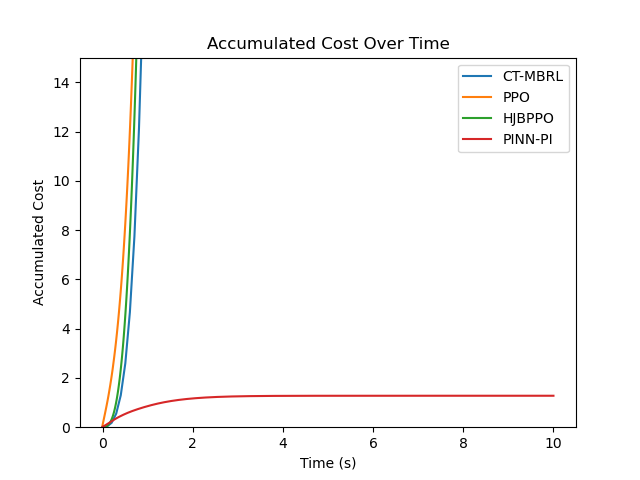}}
        \caption{Plots of accumulated costs over time for the four environments}\label{fig:acc_cost}
\end{figure}

Furthermore, Figure \ref{fig:trajectories} show simulated trajectories from different initial conditions under the learned optimal controllers. 

\begin{figure}[H]
    \centering
     \subfigure[Inverted Pendulum]{\includegraphics[width=0.4\linewidth]{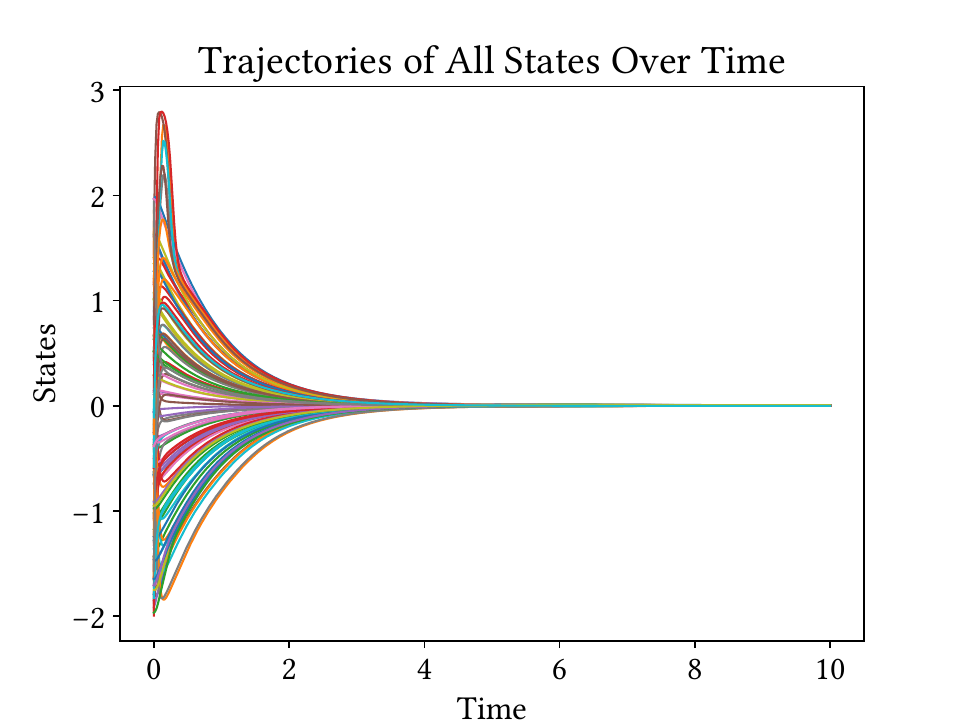}}
     \subfigure[Cartpole]{\includegraphics[width=0.4\linewidth]{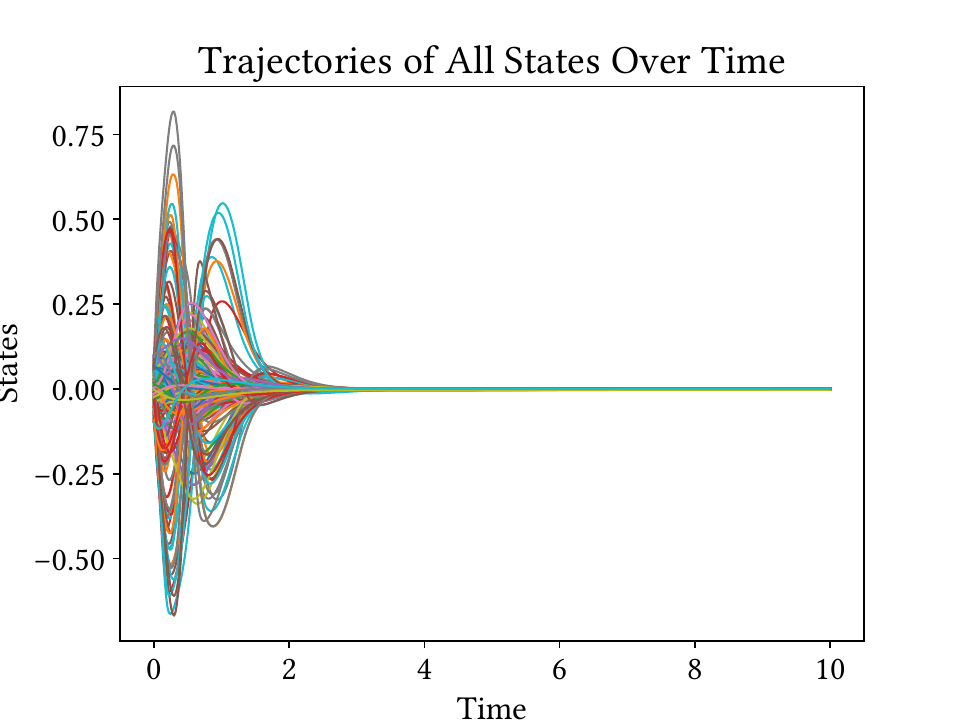}}
     \subfigure[2D Quadrotor]{\includegraphics[width=0.4\linewidth]{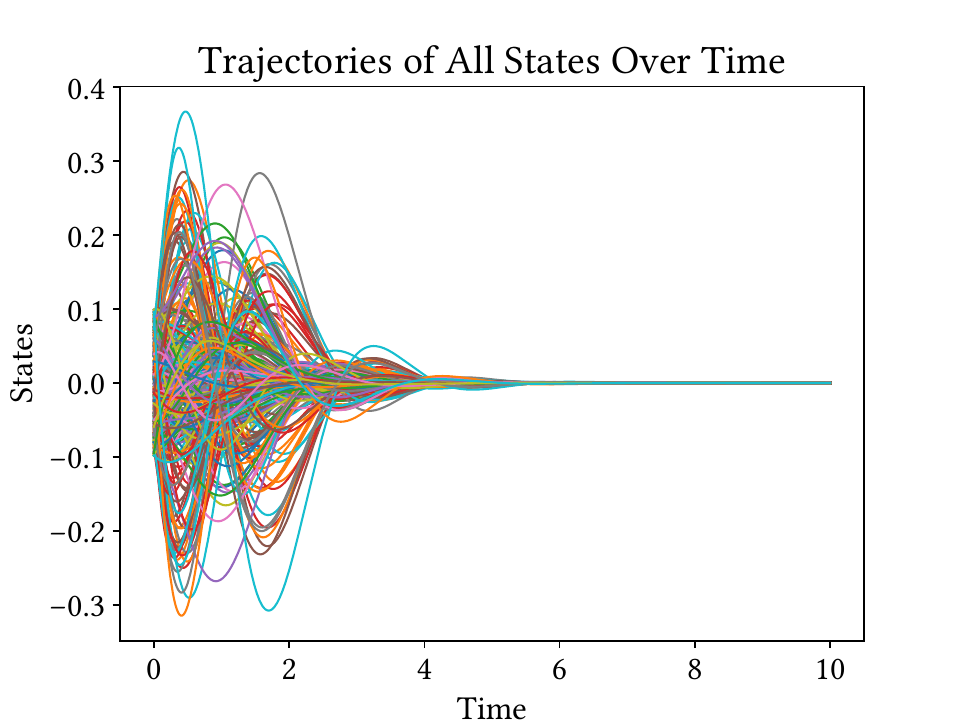}}
     \subfigure[3D Quadrotor]{\includegraphics[width=0.4\linewidth]{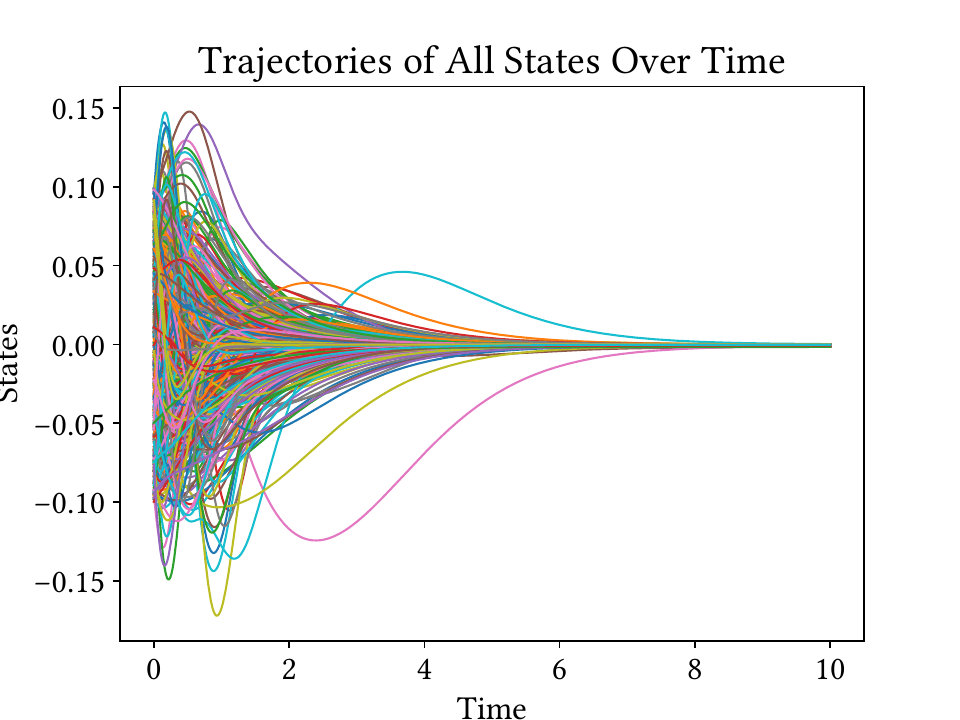}}
        \caption{Plots of trajectories starting from different initial conditions under the optimal controller learned using PINN-PI for the four environments. All trajectories converge to the origin.}\label{fig:trajectories}
\end{figure}

\subsection{Additional case studies}

\subsubsection{1D bilinear system}\label{sec:bilinear}

Recall the bilinear scalar problem $\dot x = xu$, with $Q(x) = x^2$ and $R = 1$. The optimal value function is $V(x) = 2\abs{x}$, which fails to be differentiable at $x = 0$. We show that ELM-PI can converge to the optimal value function. We choose the activation function to be ReLU and set the bias term to zero. ELM-PI achieves 1E-16 accuracy with $m \geq 3$. While this is a simple example, it demonstrates that ELM-PI can potentially achieve arbitrary accuracy, provided that the neural network is capable of approximating the value function. A plot of the obtained value function after 10 iterations is included in Figure \ref{fig:bilinear}. 

\begin{figure}[h]
    \centering
        \includegraphics[width=0.47\textwidth]{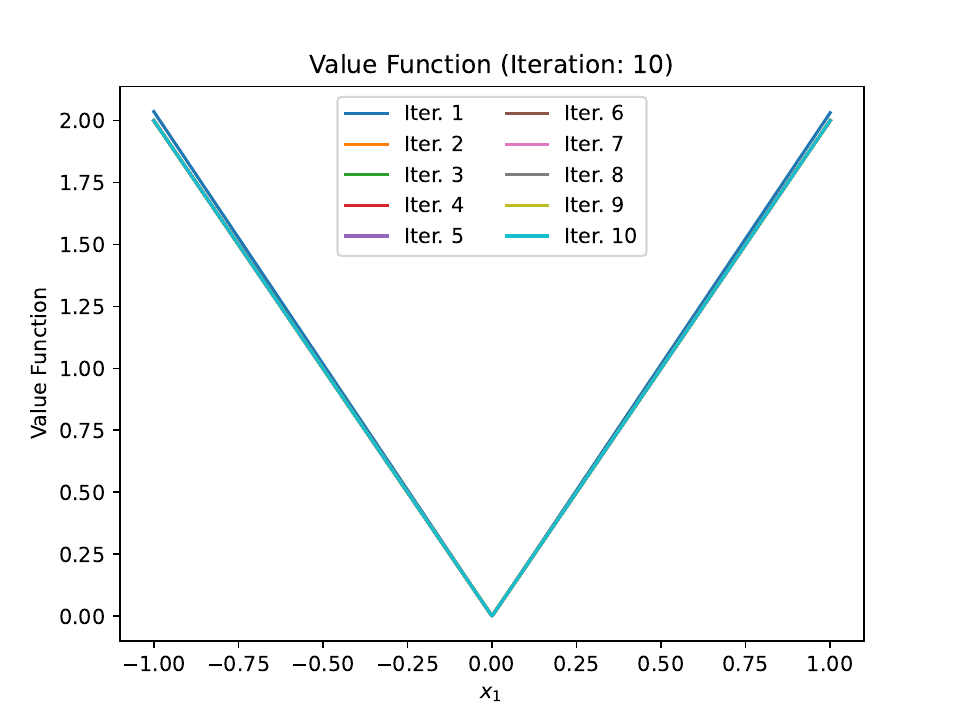}
    \caption{ELM-PI on the bilinear example with $m=10$. The error between the optimal value and computed value function is within 1E-15 after two iterations.}
    \label{fig:bilinear}
\end{figure}

\subsubsection{Lorenz system}

We consider the stabilization of a chaotic system 
\begin{equation}
    \begin{aligned}
        \dot x_1 &= -10x_1 + 10x_2 + u\\
        \dot x_2 & = 28x_1 - x_2 -x_1x_2\\
        \dot x_3 & = -\frac83 x_2 + x_1x_2
    \end{aligned}
\end{equation}
Without control, the origin is a saddle equilibrium point. We would like to stabilize the system to the origin via policy iteration. We first run ELM-PI with $m=100$, $m=200$, $m=400$, and $m=800$. The computational times are 0.68, 1.55, 7.03, and 46.84 seconds, respectively. In comparison, SGA with polynomial bases of order 2, 4, and 6 takes 6.63, 69.15, and 893.28 seconds. It can be seen from numerical simulations that ELM-PI with $m=400$ and $m=800$ leads to stabilizing controllers, whereas $m=100$ and $m=200$ give unstable controllers. Figure \ref{fig:lorenz_trajectories} depicts 10 simulated closed-loop trajectories under the controllers returned by ELM-PI with $m=400$ and $m=800$. The performance of the controllers returned by ELM-PI and the initial controller obtained from eigenvalue assignment for the linearized system are shown in Figure \ref{fig:lorenz_cost_compare} through simulated trajectories. 

\begin{figure}[h]
    \centering
        \includegraphics[width=0.47\textwidth]{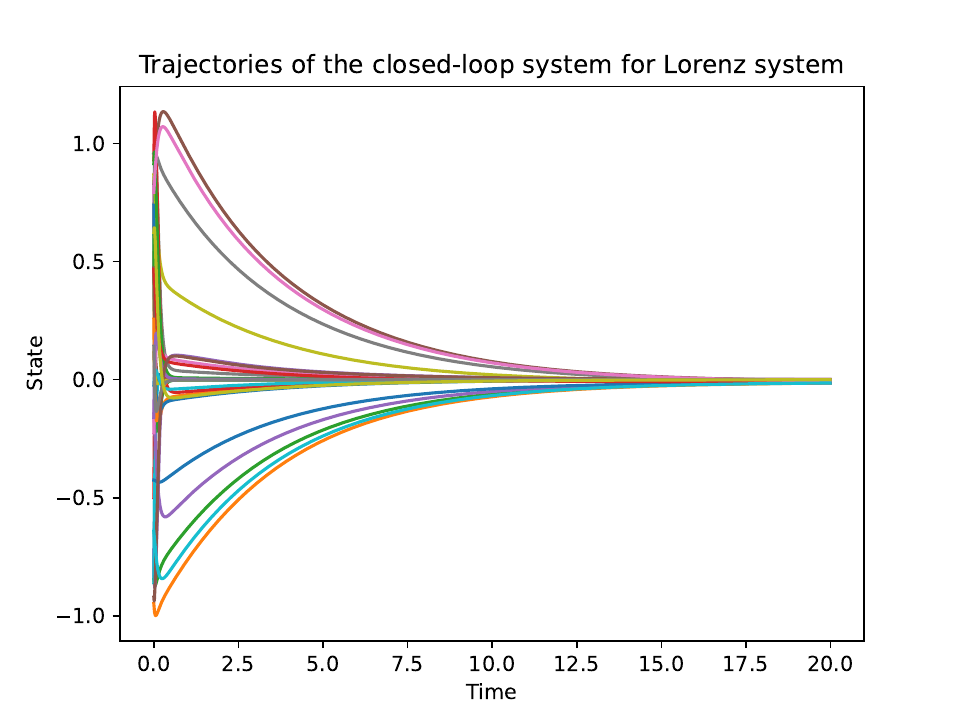}
        \includegraphics[width=0.47\textwidth]{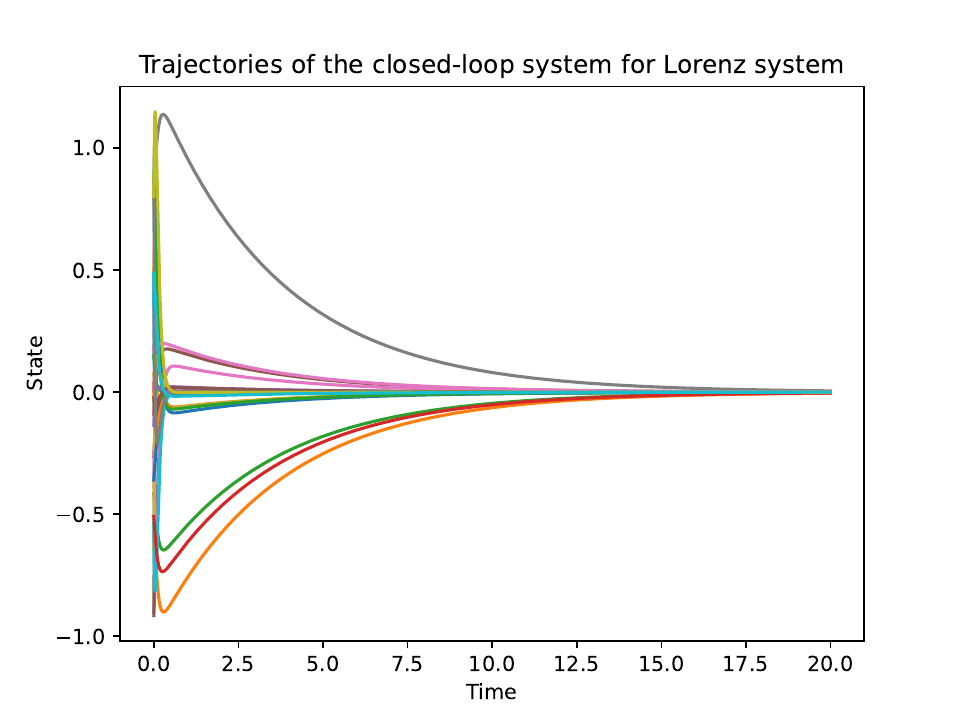}
    \caption{ELM-PI with $m=400$ and $m=800$ for the Lorenz system: the left panel shows the closed-loop trajectories under the controller returned by ELM-PI with $m=400$, and the right panel for $m=800$.}
    \label{fig:lorenz_cost_compare}
\end{figure}

\begin{figure}[h]
    \centering
        \includegraphics[width=0.47\textwidth]{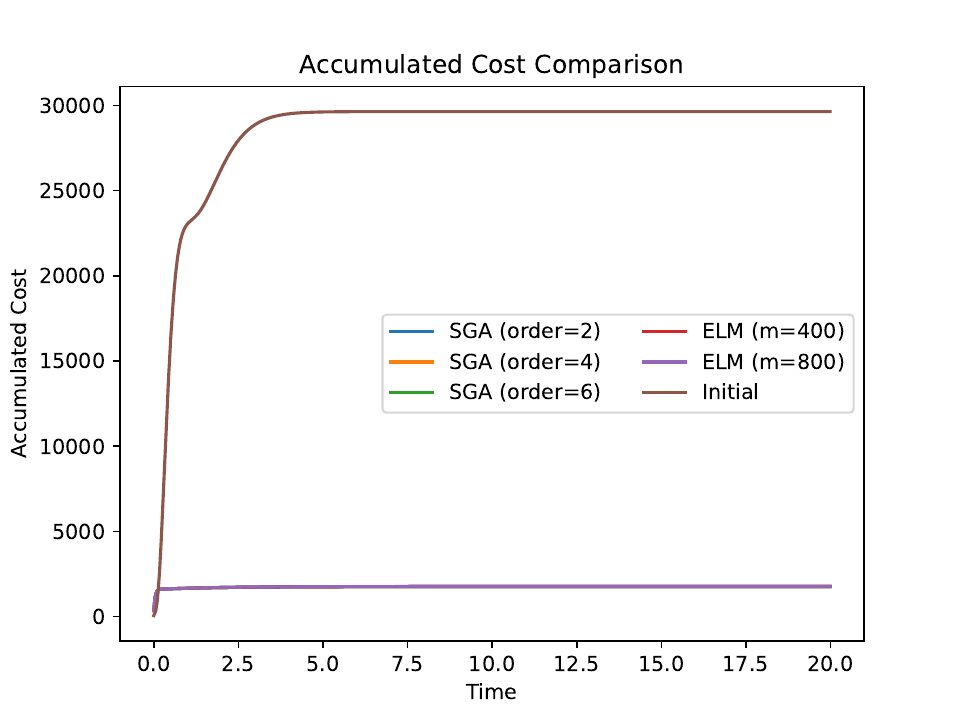}
        \includegraphics[width=0.47\textwidth]{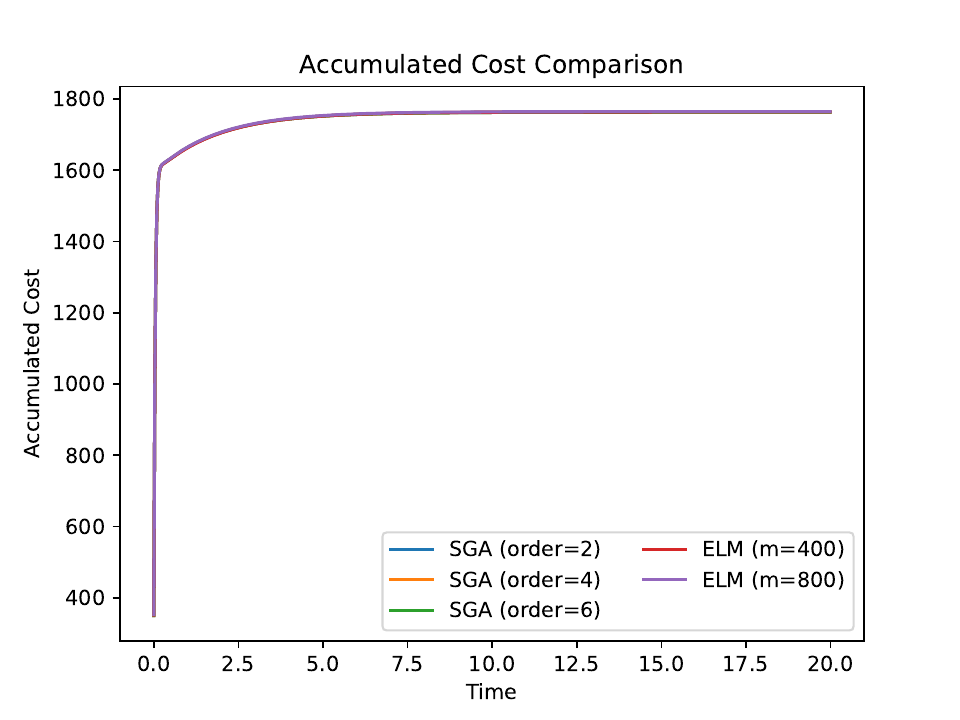}
    \caption{Simulated costs for ELM-PI with $m=400$ and $m=800$ for the Lorenz system, compared with the performance of the initial controller obtained from eigenvalue assignment for the linearized system and SGA with polynomial bases of orders 2, 4, and 6. It can be seen that both ELM and SGA achieve almost identical cost, while the computational time required by ELM-PI is considerably less.}
    \label{fig:lorenz_trajectories}
\end{figure}

\section{Limitations and future work}\label{sec:future}

We discuss a few limitations of the proposed work and potential future work in this section. 

\begin{itemize}
    \item \textbf{Convergence analysis:} In the convergence analysis, we established that both the exact PI and approximate PI can converge to the true optimal value and controller, provided that the training error can be made arbitrarily small and the training set forms a dense subset of the domain. While this is theoretically interesting, the results do not offer convergence rates or finite sample approximation guarantees. This could be an interesting topic for future research.

    \item \textbf{Initial controller and training over larger domains:} One of the main drawbacks of PI is that it requires a stabilizing controller to begin with. On the other hand, we noticed that both PINN-PI and current RL algorithms also struggle to learn a stabilizing controller over larger domains. In a small region around the equilibrium point, it is always possible to use a linear controller. Since PI requires a controlled invariant set to train the subsequent value and control functions, an interesting topic for future investigation is how to combine controllers that can guarantee to reach a small region of attraction, patched together with a local stabilizing controller, to offer opportunities for training PINN-PI over a larger domain. 

    \item \textbf{Verification:} Formal verification remains challenging for high-dimensional value functions. This difficulty seems unavoidable when computing general optimal value functions. However, there might be ways to circumvent this issue by designing cost metrics that encourage compositional controllers and value functions, or one can deliberately seek compositionally verifiable controllers that are suboptimal, yet still provide satisfactory performance with stability guarantees. We also remarked that the value functions computed with ELM-PI seem harder to verify than those computed with PINN-PI. This may be due to gradient descent implicitly regularizing the functions. Future work can investigate this discrepancy in more detail. One can also incorporate probabilistic guarantees with randomized algorithms for testing.

\end{itemize}

\newpage

\end{document}